\documentclass[a4paper,runningheads,envcountsect,envcountsame]{llncs}
\usepackage[numbers]{natbib}
\usepackage{times}
\usepackage{amsmath,amsfonts,bm}

\usepackage{color,graphicx,float}
\usepackage{subcaption}
\usepackage{algorithm}
\usepackage{algpseudocode}
\usepackage{paralist}

\usepackage{thmtools}
\usepackage{thm-restate}

\usepackage{cleveref}

\declaretheorem[name=Theorem,numberwithin=section]{thm}
\declaretheorem[name=Lemma,numberwithin=section]{lem}
\renewenvironment{theorem}{\begin{thm}}{\end{thm}}
\renewenvironment{lemma}{\begin{lem}}{\end{lem}}

\usepackage{paralist}
\newcommand{\sv}{\varphi}
\renewcommand{\cal}[1]{\mathcal{#1}}
\DeclareMathOperator*{\pr}{Pr}
\newcommand{\commentout}[1]{}
\DeclareMathOperator*{\Ex}{\mathbb{E}}
\newcommand{\R}{\mathbb{R}}

\DeclareMathOperator*{\E}{\mathbb{E}}

\newcommand{\var}{\mathrm{Var}}
\renewcommand{\vec}[1]{\bm{#1}}

\newcommand{\bin}{\mathrm{B}}
\newcommand{\nor}{\mathcal{N}}
\newcommand{\eps}{\varepsilon}
\DeclareMathOperator{\pois}{Pois}
\providecommand{\shapley}{\sv}

\providecommand{\pseq}{a}
\providecommand{\pseqall}{A}

\DeclareMathOperator{\Sym}{Sym}
\providecommand{\perms}[1]{\Sym_{#1}}
\providecommand{\setbinom}[2]{\genfrac{[}{]}{0pt}{}{#1}{#2}}
\providecommand{\inR}{\in_R}
\providecommand{\teps}{t_{\eps}}

\title{Power Distribution in Randomized Weighted Voting: the Effects of the Quota}
\author{
Joel Oren\inst{1}\and
Yuval Filmus\inst{2}\and
Yair Zick\inst{3}\and
Yoram Bachrach\inst{4}
}
\institute{
Department of Computer Science,\\
University of Toronto, Canada\\
\email{oren@toronto.edu}
\and
Institute of Advanced Study\\ 
Princeton University, USA\\
\email{yfilmus@ias.edu}
\and
School of Computer Science\\
Carnegie-Mellon University, USA\\
\email{yairzick@cs.cmu.edu}
\and
Microsoft Research, UK\\
\email{yobach@microsoft.com}
}
\date{}
\newif\ifcomments
\commentsfalse
\ifcomments
\newcommand{\kibitz}[2]{{\color{#1}{#2}}}
\else
\newcommand{\kibitz}[2]{}
\fi
\newcommand{\jo}[1]{\kibitz{magenta}{[Joel: #1]}}

\newcommand{\yz}[1]{\kibitz{red}{[Yair: #1]}}
\newcommand{\yf}[1]{\kibitz{green}{[Yuval: #1]}}

\begin{document}
\maketitle

\begin{abstract}
We study the Shapley value in weighted voting games. The Shapley value
has been used as an index for measuring the power of individual agents
in decision-making bodies and 
political organizations, where decisions are made by a majority vote
process. We characterize the impact of changing the quota (i.e., the
minimum number of seats in the parliament that are required to form a
coalition) on the Shapley values of the agents.
Contrary to previous studies, which assumed that the agent weights
(corresponding to the size of a caucus or a political party) are
fixed, we analyze new domains in which the weights are stochastically generated, modeling, for
example, elections processes.


We examine a natural weight generation process: the Balls and Bins
model, with uniform as well as exponentially decaying
probabilities. 
We also analyze weights that admit a super-increasing sequence,
answering several open questions pertaining to the Shapley values in such games.
\end{abstract}
\section{Introduction}
\label{sec:intro}
Weighted voting is a common method for making group decisions. This is the method used in parliaments: one can think of the political parties in a parliament as \emph{weighted agents}, where an agent's weight is the number of seats it holds in the parliament. 

Power dynamics in electoral systems have been the focus of academic study for several decades. One important observation is that the weight of a party is not necessarily equal to its electoral power. For example, consider a parliament that has three parties, two with $50$ seats, and one with $20$ seats. Assuming that a majority of the votes is required in order to pass a bill, all three parties have the same decision-making power: no single party can pass a bill on its own, whereas any two parties can. This contrasts the fact that one of the parties has significantly less weight than the other two. 
One of the most prominent measures of voting power is the Shapley-Shubik power index (also referred to as the Shapley value); it has played a central role in the analysis of real-life voting systems, such as the US electoral college~\cite{mann60sv, mann62sv}, the EU council of members~\cite{leech02vote,slom2006penrose,felsen98book}, and the UN security council~\cite{shapleyvalue}. 

Empirical studies of weighted voting present an interesting phenomenon: changes to the quota (i.e., the number of votes required in order to pass a bill, also called the threshold) can dramatically affect agent voting power. Changes to the quota have been proposed as a way to correct power imbalance in the EU council of members~\cite{slom2006penrose}; this is because quota changes are perceived as a preferable alternative to changes to agent weights (as is proposed by~\cite{penrose46}), and were thus argued for in~\cite{slom2006penrose,leech03quota}. 

The objective of this paper is to study the effects of changes to the quota on electoral power as measured by the Shapley-Shubik power index. Previous analytical studies of power indices as a function of the quota \jo{commented-out citations ---this should appear in the previous work, as it muddles the intro.~\cite{zick2011sv,zick13var, zuck12manip}} have mostly focused on the following question: given a set of weights, what would be the effect of changes to the quota on voting power? 

As mentioned, the effect that changing the quota has on the Shapley value has been studied to some extent. However, as these studies show, relatively little can be said about these effects in general. 
Instead of studying arbitrary weight vectors, we assume that weights are sampled from certain natural distributions. 
Modeling a parliamentary election process, we think of voters as casting their ballots according to a prescribed distribution, that determines the number of seats each party will hold. 
Using natural weight generation processes, we analyze the expected behavior of the Shapley value as a function of the quota. For example, some of our results show that even when weights are likely to be very similar, some choices of a quota will cause significant differences in voting power. 

\paragraph{Our Contributions}
\label{sec:contrib}



Our work focuses on the \emph{Balls and Bins} model ---a model that has received considerable recent attention in the computer science community~\cite{Mitzenmacher00thepower,upfalProbBook,raab98}. 
Informally, in this iterative process, in each round a ball is thrown into one of several bins according to a fixed probability distribution. 

In Section~\ref{sec:bbuniform}, we study a simple model, where each ball lands in one of the $n$ bins uniformly at random. 
We identify a repetitive fluctuation pattern in the Shapley values, with cycles of length $\frac mn$. We show that if the quota is sufficiently bounded away from the borders of its length-$\frac{m}{n}$ cycle, then the Shapley values of all agents are likely to be very close to each other. On the other hand, we show that due to noise effects, when the quota is situated close enough to small multiples of $\frac mn$, the highest Shapley value can be roughly double than that of the smallest one. In other words, even if one expects that candidate weights are identical with high probability, choosing a quota near a multiple of $\frac mn$ may result in a great difference between Shapley values. 

To complement our findings for the uniform case, in Section~\ref{sec:balls_bins_exponential} we consider the case in which the probabilities decay exponentially, with a decay factor no larger than $1/2$. We show that analyzing this case essentially boils down to characterizing the Shapley values in a game where weights are a super-increasing sequence (Section~\ref{sec:super-increasing}). Our results significantly strengthen previous results obtained for this case in~\cite{zuck12manip}: we fully characterize the Shapley value as a function of the quota for the super-increasing case. In addition to giving a closed-form formula for the Shapley values, we also provide conditions for the equality of consecutive agents.

\subsection{Related Work}
\label{sec:related}
Weighted voting games (WVGs) have been studied extensively, two classic power measures proposed by \cite{banzhaf} and by \cite{shapleyvalue,shapleyshubik} being the main object of analysis (see \cite{felsen98book}, \cite{gtbook} and \cite{coopbook} for expositions). 
From an economic point of view, the appeal of the Shapley value is that it is the {\em only} division rule that satisfies certain desirable axioms~\cite{shapleyvalue}.
Computing Shapley values in WVGs has also been the focus of several studies: power indices have been shown to be computationally intractable (see~\cite{compcoopbook} for a detailed overview), but easily approximable. 
Randomized sampling has been employed in order to efficiently approximate the Shapley value, with the earliest examples of this technique appears in~\cite{mann60sv}, with subsequent analysis in~\cite{bachrach10approx,fatima2008linear}. However, this type of analysis employs the inherent probabilistic nature of power indices, rather than inducing randomness in the weighted voting game itself.

If one makes no assumptions on weight distributions, very little can be said about the effects of the quota on WVGs.
Indeed, as demonstrated in~\cite{zick2011sv,zick13var,zuck12manip}, power measures are highly sensitive to varying quota values. While~\cite{zick13var} presents some preliminary results on the effects of the quota when weights are sampled from a given distribution, our work takes a more principled approach to the matter.

Several works have studied the effects of randomization on weighted voting games from a theoretical, computational and empirical perspective. The earliest study of randomization and its effects on voting power is due to \cite{penrose46}, who shows that the Banzhaf power index scales as the square root of players' weight when weights are drawn from bounded distributions.\footnote{The results shown by Penrose predate the work by Banzhaf, but can be applied directly to his work; see~\cite{felsen05voting} for details.} \cite{lindner2004wvg} shows certain convergence results for power indices, when players are sampled from some distributions; \cite{jelnov2012sv} show that when weights are sampled from the uniform distribution, the expected Shapley value of a player is proportional to its weight, assuming that the quota is 50\%. \cite{zick13var} considers a model where the quota is sampled from a uniform distribution, and bounds the variance of the Shapley value in this setting, both for general weights and for weights sampled from certain distributions.
The effects of changes to the quota have also been studied empirically, mostly in the context of the EU council of members~\cite{leech03quota,leech02vote,slom2006penrose}.

\section{Preliminaries}
\label{sec:prelim}
\paragraph{General notation}
Given a vector $\vec x \in \R^n$ and a set $S\subseteq \{1,\dots,n\}$, let $x(S) = \sum_{i \in S}x_i$. For a random variable $X$, we let $\E[X]$ be its expectation, and $\var[X]$ be its variance.
For a set $S$, we denote by $\setbinom{S}{k}$ the collection of subsets of $S$ of cardinality $k$. The notation $T \inR \setbinom{S}{k}$ means that the set $T$ is chosen uniformly at random from $\setbinom{S}{k}$.
We let $\bin(n,p)$ denote the binomial distribution with $n$ trials and success probability $p$.
We let $\nor(\mu,\sigma^2)$ denote the normal distribution with mean $\mu$ and variance $\sigma^2$.
We let $U(a,b)$ denote the uniform distribution on the interval $[a,b]$.

We let $O_p( \cdot )$ denote the usual big-O notation, conditioned on a fixed value of $p$. In other words, having $f(n)=O_p(g(n))$ means that there exist functions $K(\cdot), N(\cdot)$, such that for $n \geq N(p)$,  $f(n) \leq K(p) \cdot g(n)$.

Finally, for a distribution $D$ over $\R$, and some event $\mathcal{E}$, we simplify our notation by letting $\pr [\cal{E}(D)] = \pr_{x \sim D} [ \mathcal{E}(x) ]$. For example, for $a > 0$, we can write $\pr [\bin(n,p) \leq a ] = \pr_{x \sim \bin(n,p)}[x \leq a]$.

\paragraph{Weighted voting games} A {\em weighted voting game} (WVG) is given by a set of agents $N = \{1,\dots,n\}$, where each agent $i \in N$ has a positive weight $w_i$, and a {\em quota} (or {\em threshold}) $q$. Unless otherwise specified, we assume that the weights are arranged in non-decreasing order, $w_1\le \dots\le w_n$.
For a subset of agents $S \subseteq N$, we define $w(S) = \sum_{i \in S} w_i$.

A subset of agents $S \subseteq N$ is called {\em winning} (has value $1$) if $w(S) \ge q$ and is called losing (has value $0$) otherwise. 

\paragraph{The Shapley value} Let $\perms{n}$ be the set of all permutations of $N$. Given some permutation $\sigma \in \perms{n}$ and an agent $i \in N$, we let $P_i(\sigma) = \{j \in N : \sigma(j) < \sigma(i)\}$; $P_i(\sigma)$ is called the set of $i$'s predecessors in $\sigma$. Let us write $m_i(S)$ to be $v(S\cup \{i\}) - v(S)$; in other words, $m_i(S) =1$ if and only if $v(S) = 0$ but $v(S\cup\{i\}) = 1$. If $m_i(S) = 1$, we say that $i$ is {\em pivotal} for $S$; similarly, we write $m_i(\sigma) = m_i(P_i(\sigma))$, and say that $i$ is pivotal for $\sigma \in \perms{n}$ if $i$ is pivotal for $P_i(\sigma)$. The Shapley-Shubik power index (often referred to as as the Shapley value in the context of WVG's) is simply the probability that $i$ is pivotal for a permutation $\sigma \in \perms{n}$ selected uniformly at random. More explicitly,
$$\sv_i  = \frac{1}{n!}\sum_{\sigma \in \perms{n}}m_i(\sigma).$$
Since $\sigma^{-1}(i)$ is distributed uniformly when $\sigma$ is chosen at random from $\perms{n}$, we also have the alternative formula
\begin{equation} \label{eq:shapley-alt-formula}
\sv_i = \frac{1}{n} \sum_{\ell=0}^{n-1} \E_{S \inR \setbinom{N \setminus \{i\}}{\ell}} m_i(S).
\end{equation}

\paragraph{Properties of the Shapley value} For WVGs, it is not hard to show that $w_i \leq w_j$ implies $\sv_i \leq \sv_j$, and so if the weights are arranged in non-decreasing order, the minimal Shapley value is $\sv_1$ and the maximal one is $\sv_n$. Another useful property that follows immediately from the definitions is that $\sum_{i \in N} \sv_i = 1$, assuming $0 < q \leq \sum_{i \in N} w_i$. When we want to emphasize the role of the quota $q$, we will think of the Shapley values as functions of $q$: $\sv_i(q)$.

\section{An Overview of our Results}\label{sec:overview}
We begin by briefly presenting our three major contributions. 
\paragraph{The Balls and Bins Distribution: the Uniform Case}
In this work, we study the effects of the quota on agents' voting power, when agent weights are sampled from the balls and bins distribution. 
This distribution is appealing, as it can naturally model election dynamics under plurality voting: consider an election where $m$ voters vote for $n$ parties; the weight of each party is determined by the number of votes it receives. If we assume that each voter will vote for party $i$ with probability $p_i$, party seats are distributed according to the balls and bins distribution with the probability vector $\vec p$. 
We first study the case where balls are thrown into bins uniformly at random; that is, voters choose parties uniformly at random (the {\em impartial culture} assumption). 

\begin{figure}[ht!]%
\centering
\begin{subfigure}[t]{0.44\textwidth}
\includegraphics[width = \columnwidth]{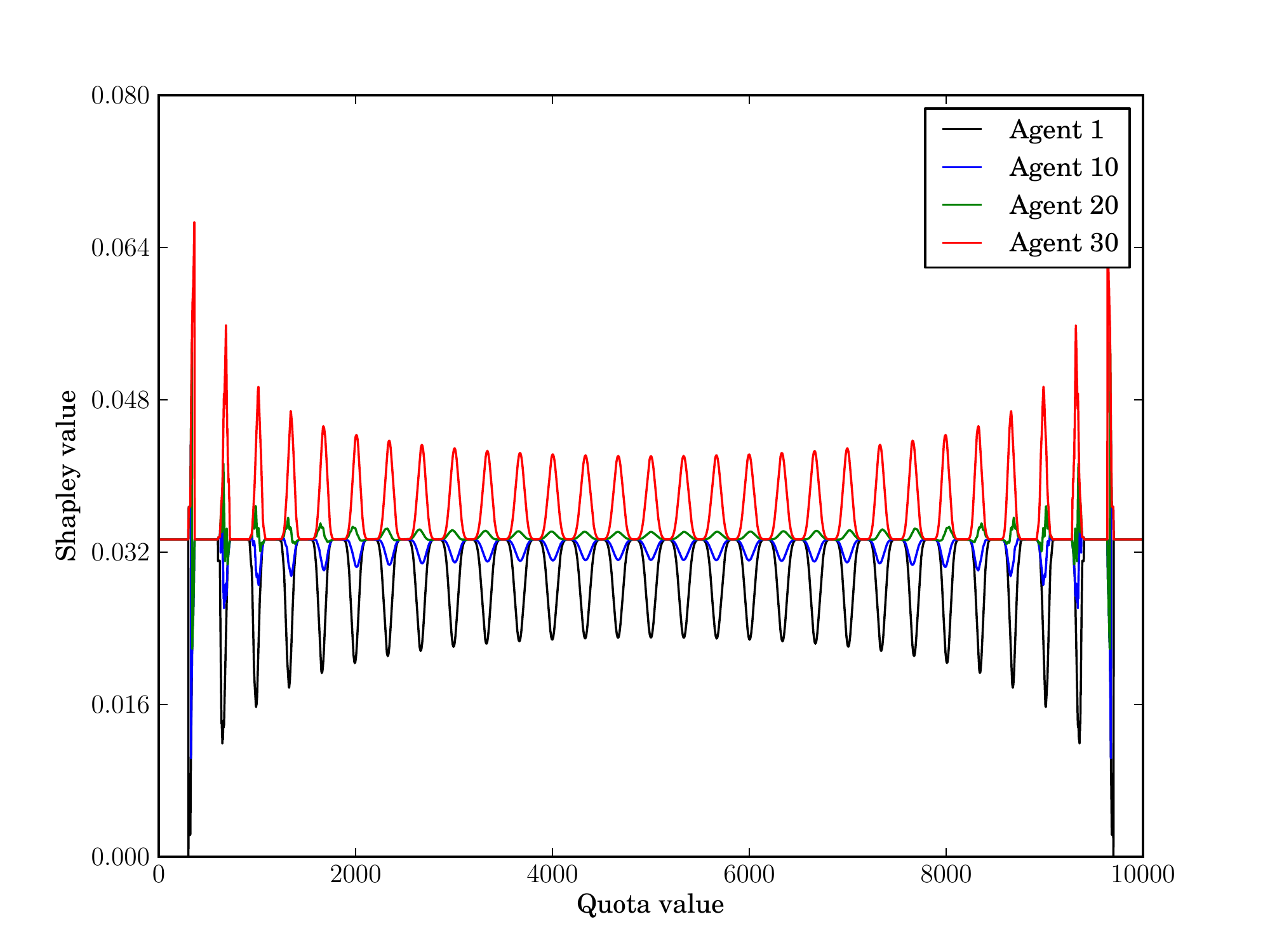}%
\caption{The Shapley values of agents 1, 10, 20 and 30 in a 30-agent WVG where weights were drawn from a balls and bins distribution with $m = 10000$ balls.}%
\label{fig:sv_30_players_10000_balls}
\end{subfigure}
\hspace{0.02\textwidth}
\begin{subfigure}[t]{0.44\textwidth}
\includegraphics[width = \columnwidth]{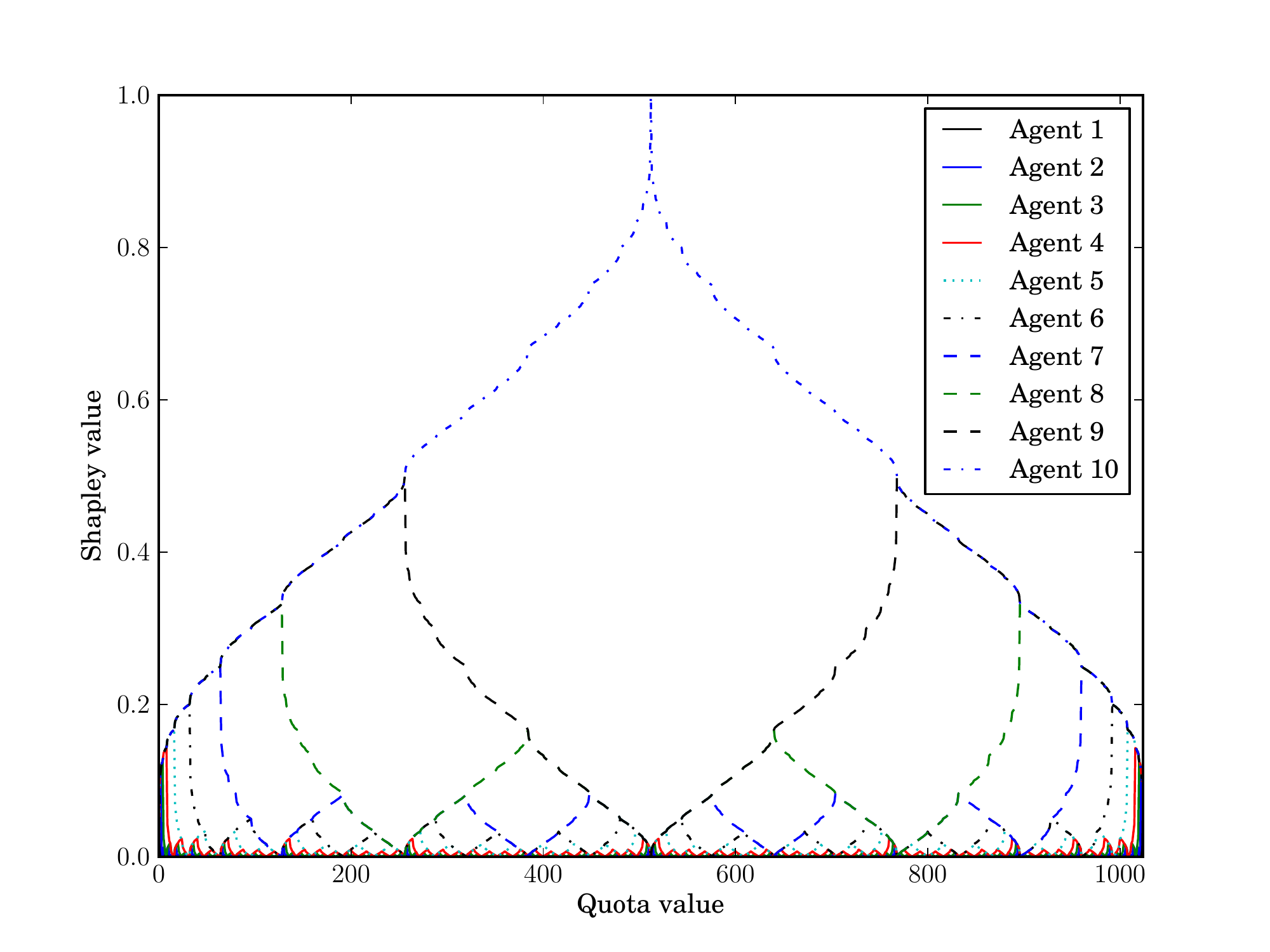}%
\caption{The Shapley values as a function of the quota in a 10-agent game where agent $i$'s weight is $2^{i-1}$.}%
\label{fig:sv_pows_of_two_10}%
\end{subfigure}
\label{fig:bigfig}
\end{figure}

When weights are drawn from a uniform balls and bins distribution with $m$ balls and $n$ bins, Shapley values follow a rather curious fluctuation pattern as the quota varies (see Figure~\ref{fig:sv_30_players_10000_balls}). Note that the fluctuation is quite regular, with power disparity occurring at regular intervals (these intervals are of length $\frac mn$). Our first result (Theorem~\ref{thm:plurality_positive}) shows that when we select a quota that is sufficiently far from an integer multiple of $\frac mn$, all agents' Shapley values tend to be the same. When the quota is an integer multiple of $\frac mn$, we distinguish between two cases; when the quota is far from the 50\% mark, power disparity is likely to occur, with the weakest agent's voting power sinking to less than half that of the strongest (Theorem~\ref{thm:min_shap_small_l}). However, disparity is mitigated when the quota is near the 50\% mark (Theorem~\ref{thm:min_shap_large_l}). These results indicate that even if weights are likely to be similar (as is the case for the uniform balls and bins distribution), power disparity is likely in certain quotas.

\paragraph{The Balls and Bins Distribution: the Exponential Case}
In Section~\ref{sec:balls_bins_exponential}, we explore the case where the voting probabilities are exponentially increasing, i.e. $\frac{p_i}{p_{i+1}} = \rho$ for some fixed constant $0 < \rho < \frac12$. In this case, we show (Theorem~\ref{thm:super-increasing3}) that agents' weights are very likely to be {\em super-increasing} (super-increasing weights were first studied by~\cite{zuck12manip}).  
Thus, in order to understand the expected behavior of voting power as a function of the quota in the exponential setting, it suffices to characterize the Shapley value for weighted voting games with super increasing weights.

\paragraph{Super-Increasing Weights}
Following the crucial observation made in Theorem~\ref{thm:super-increasing3}, we complete characterize voting power in WVGs with super-increasing weights in Section~\ref{sec:super-increasing}. First we show that in order to compute the Shapley value of an agent under a super-increasing sequence, it suffices to know his Shapley value when weights are powers of 2 (Lemma~\ref{lem:super-increasing-P}). This connection leads to a closed-form formula for the Shapley value when weights are super-increasing (Theorem~\ref{thm:super-increasing-formula}). 
Employing our formula, we are able to derive some interesting properties of the Shapley values as functions of the quota for super-increasing sequences. These results generalize those found in~\cite{zuck12manip}, providing a clear understanding of the mechanics of power distribution and the quota for the case of super-increasing sequences. 

We conclude our study with an interesting analysis of $\sv_i(q)$ when weights are finite prefixes of the sequence $\left(2^m\right)_{m = 0}^\infty$. The analysis explains in many ways the fractal shape of $\sv_i(q)$ when weights are powers of two, and shows when voting power will increase or decrease. 


\section{The Balls and Bins Distribution: the Uniform Case}
\label{sec:bbuniform}
We now consider a generative stochastic process called the Balls and Bins process. In its most general form, given a set of $n$ bins and a distribution represented by a vector $\mathbf{p} \in [0,1]^n$ such that $\sum_{i=1}^n p_i = 1$, the process unfolds in $m$ steps. At every step, a ball is thrown into one of the bins based on the probability vector $\vec p$. The resulting weights are then sorted in non-decreasing order $w_1 \leq \dots \leq w_n$.


We begin our study of the balls and bins process by considering the most commonly studied version of the balls and bins model, in which each ball is thrown into one of the bins with equal probability, i.e., $p_i=1/n$, for all $i \in N$.


As Figure~\ref{fig:sv_30_players_10000_balls} shows for the case of $n=30$, the behavior of the Shapley values demonstrates an almost perfect cyclic pattern, with intervals of length $m/n$. As can be seen in the figure, for quota values that are sufficiently distant from the interval endpoints, all of the Shapley values tend to be equal (as the Shapley values of the highest and lowest agents are equal in these regions).
As the number of balls grows, all of the bins tend to have nearly the same number of balls in them; however, low weight discrepancy does not immediately translate to low power discrepancy: we can guarantee nearly equal voting power in some quotas, but not in others. 

We begin by providing a formula for the differences between two Shapley values. 
\begin{restatable}{lem}{binomialformula} \label{lem:binomial-formula}
For all agents $i,j \in N$,
\[|\sv_j - \sv_i| = \frac{1}{n-1} \sum_{\ell=0}^{n-2} \Pr_{S \in_R \setbinom{N \setminus \{i,j\}}{\ell}}[q - \max(w_i,w_j) \leq w(S) < q - \min(w_i,w_j)].\]
\end{restatable}

We now give a theoretical justification for the near-identity of Shapley values for quotas that are well-away from integer multiples of $\frac mn$.
\begin{restatable}{thm}{pluralitypositive}
\label{thm:plurality_positive}
Let $M = \frac{m}{3n^3}$.
Suppose that $|q-\frac{\ell m}{n}| > \frac{1}{\sqrt{M}} \frac{m}{n}$ for all integers $\ell$. Then with probability $1 - 2(\frac{2}{e})^n$, all Shapley values are equal to $1/n$.
\end{restatable}

The idea of the proof is the following. Suppose that $w_i \leq w_j$. According to Lemma~\ref{lem:binomial-formula}, $\sv_i \neq \sv_j$ only if for some $S \subseteq N \setminus \{i,j\}$, we have $q - w_j \leq w(S) < q - w_i$.
For a fixed  set of agents $|S| \in \setbinom{N \setminus \{i,j\}}{k}$ we have $S \sim \bin(m,k/n)$ --- as each ball enters into one of the bins corresponding to $S$ with probability $k/n$. As a result, $w(S)$ is concentrated around the mean $km/n$. On the other hand $q - w_j,q - w_i \approx q - \frac{m}{n}$. Therefore, if $q$ is far away from $\frac{(k+1)m}{n}$ for all $0 \leq k \leq n-2$, then the probability that $q - w_j \leq w(S) < q - w_i$ is very small. The details can be found in the appendix.

Returning to our voting setting, the interpretation of Theorem~\ref{thm:plurality_positive} is that if the voter population is much larger than the number of candidates, and the votes are assumed to be cast uniformly at random (i.e., a totally neutral distribution of preferences), then choosing a quota that is well away from a multiple of $\frac mn$, will most probably lead to an even distribution of power among the elected representatives (e.g., political parties). 

\subsection{How Weak Can the Weakest Agent Get in the Uniform Case?} \label{sec:bbuniform-weak}
As Theorem~\ref{thm:plurality_positive} demonstrates, if the quota is sufficiently bounded away from any integral multiple of $\frac mn$, then the distribution of power tends to be even among the agents. When the quota is close to an integer multiple of $\frac mn$, it may very well be that the resulting weighted voting game may not display such an even distribution of power, as a result of weight differences, as a result of the intrinsic ``noise'' of the process. Figure~\ref{fig:sv_30_players_10000_balls} provides an empirical validation of this intuition. Motivated by these observations, we now proceed to study the expected Shapley value of the weakest agent, $\sv_1$ (recall that we assume that the weights are given in non-decreasing order).

We now present two contrasting results. Let $q = \ell \cdot \frac{m}{n}$, for an integer $\ell$. When $\ell = o(\log n)$, we show that the expected minimal Shapley value is roughly $\frac{1}{2n}$, and so it is at least half the maximal Shapley value (in expectation).

\begin{theorem}
\label{thm:min_shap_small_l}
  Let $q=\ell \cdot \frac{m}{n}$ for some integer $\ell = o(\log n)$. For $m = \Omega(n^3\log n)$, $\E [\sv_1] = \frac{1}{2n} + o(\frac{1}{n})$.
\end{theorem}

In contrast, when $\ell = \Omega(n)$, this effect disappears.

\begin{theorem} \label{thm:min_shap_large_l}
Let $q = \ell \cdot \frac{m}{n}$ for $\ell \in \{1,\ldots,n\}$ such that $\gamma \leq \frac{\ell}{n} \leq 1-\gamma$ for some constant $\gamma > 0$. Then for $m = \Omega(n^3)$, 
$\E[\sv_1] \geq \frac{1}{n} - O_\gamma \left(\sqrt{\frac{\log n}{n^3}}\right).$
\end{theorem}

The idea behind the proof of both theorems is the formula for $\sv_1$ given in Lemma~\ref{lem:min_shapley_bound}. In this formula and in the rest of the section, the probabilities are taken over both the displayed variables and the choice of weights.

\begin{restatable}{lem}{minshapleybound}
\label{lem:min_shapley_bound}
  Let $q=\ell \cdot \frac{m}{n}$, where $\ell \in \{1,\ldots,n-1\}$. For $m = \Omega(n^3\log n)$,
  \[
  \E [\sv_1] = \frac{1}{2(n-\ell)} - \frac{\ell}{n(n-\ell)} + \frac{1}{n-\ell} \pr_{A \inR \setbinom{N \setminus \{1\}}{\ell-1}}[w(A)+w_1 \geq q] \pm O\left(\frac{1}{n^2}\right).
  \]
\end{restatable}

The full details of the proof appear in the appendix.

In order to estimate the expression $\pr_{A \inR \setbinom{N \setminus \{1\}}{\ell-1}}[w(A)+w_1 \geq q]$, we need a good estimate for $w_1$. Such an estimate is given by the following lemma.
\begin{restatable}{lem}{minweight}\label{lem:min_weight}
With probability $1 - 2/n$, we have that $\sqrt{\frac{m\log n}{3n}} \leq \frac{m}{n} - w_1 \leq \sqrt{\frac{4m\log n}{n}}.$
\end{restatable}
We obtain this bound by applying the \emph{Poisson approximation technique} to the Balls and Bins process, which we now roughly describe. Consider the case of a random event, defined with respect to the
weight distribution induced by the process. The probability of the event can be well-approximated by the probability of an analogous event, defined with respect to $n$ \emph{i.i.d.} Poisson random variables, assuming the event is monotone in the number of balls.

We can now prove Theorem~\ref{thm:min_shap_small_l}.
\begin{proof}[of Theorem~\ref{thm:min_shap_small_l}]
Lemma~\ref{lem:prob-est} (a simple technical result proved in the appendix) shows that
\[
\pr_{A \inR \setbinom{N \setminus \{1\}}{\ell-1}}[w(A)+w_1 \geq q] \leq \frac{n}{n-\ell+1} \pr[\bin(m,\tfrac{\ell-1}{n}) \geq q-w_1].
\]
The concentration bound on $w_1$ (Lemma~\ref{lem:min_weight}) shows that with probability $1-2/n$, $q-w_1 \geq \frac{(\ell-1)m}{n} + \sqrt{\frac{m\log n}{3n}}$. Assuming this, a Chernoff bound gives
\begin{align*}
\pr[\bin(m,\tfrac{\ell-1}{n}) \geq q-w_1] &\leq \pr[\bin(m,\tfrac{\ell-1}{n}) \geq \frac{(\ell-1)m}{n} + \sqrt{\frac{m\log n}{3n}}] \leq e^{-\frac{m\log n/(3n)}{3(\ell-1)m/n}}= o(1), 
\end{align*}
using $\ell = o(\log n)$. Accounting for possible failure of the bound on $q-w_1$, we obtain
\[
\pr_{A \inR \setbinom{N \setminus \{1\}}{\ell-1}}[w(A)+w_1 \geq q] \leq \left(1-\frac{2}{n}\right) \cdot o\left(\frac{n}{n-\ell}\right) + \frac{2}{n} \cdot 1 = o(1),
\]
using $\ell = o(\log n)$.
Lemma~\ref{lem:min_shapley_bound} therefore shows that
\[ \E[\sv_1] \leq \frac{1}{2(n-\ell)} + o\left(\frac{1}{n-\ell}\right) + O\left(\frac{1}{n^2}\right) = \frac{1}{2n} + o\left(\frac{1}{n}\right), \]
since $\ell = o(\log n)$ implies $\frac{1}{n-\ell} = \frac{1}{n} + \frac{\ell}{n(n-\ell)} = \frac{1}{n} + o(\frac{1}{n})$.
Lemma~\ref{lem:min_shapley_bound} also implies a matching lower bound:
\[ \E[\sv_1] \geq \frac{1}{2(n-\ell)} - \frac{\ell}{n(n-\ell)} - O\left(\frac{1}{n^2}\right) \geq \frac{1}{2n} - o\left(\frac{1}{n}\right). \]
\qed\end{proof}

In the regime of $\ell$ addressed by Theorem~\ref{thm:min_shap_small_l}, $\pr_{A \inR \setbinom{N \setminus \{1\}}{\ell-1}}[w(A)+w_1 \geq q]$ was negligible. In contrast, in the regime of $\ell$ addressed  by  Theorem~\ref{thm:min_shap_large_l}, $\pr_{A \inR \setbinom{N \setminus \{1\}}{\ell-1}}[w(A)+w_1 \geq q] \approx 1/2$, as the following lemma, which is proved in the appendix using the Berry--Esseen theorem, shows.

\begin{restatable}{lem}{pepsbound}
  \label{lem:p_eps_bound}
  Suppose $q = \ell \frac mn$ for an integer $\ell$ satisfying $\gamma \leq \frac{\ell-1}{n} \leq 1 - \gamma$, and let
  \[ \teps = \pr_{A \inR \setbinom{N \setminus \{1\}}{\ell-1}}\left[w(A)+w_1 \geq q : w_1= \frac{m}{n}-\eps\sqrt{\frac{m\log n}{n}}\right]. \]
  Then for $m \geq 4n^3$,
\[ \teps \geq \frac{1}{2} - \frac{\eps}{2\pi\gamma} \sqrt{\frac{\log n}{n}} - \frac{1}{n}. \]
\end{restatable}

As Lemma~\ref{lem:min_weight} shows, $1/3 \leq \eps \leq 4$ with probability $1-2/n$, which explains the usefulness of this bound. We can now prove Theorem~\ref{thm:min_shap_large_l}.

\begin{proof}[of Theorem \ref{thm:min_shap_large_l}]
Lemma~\ref{lem:min_weight} shows that with probability $1-2/n$, $w_1 = \frac{m}{n}-\eps\sqrt{\frac{m\log n}{n}}$ for some $1/3 \leq \eps \leq 4$, in which regime Lemma~\ref{lem:p_eps_bound} shows that $\teps \geq \frac{1}{2} - \frac{2}{\pi\gamma} \sqrt{\frac{\log n}{n}} - \frac{1}{n}$. Accounting for the case in which $\eps$ is out of bounds,
\[
 \pr_{A \inR \setbinom{N \setminus \{1\}}{\ell-1}}[w(A)+w_1 \geq q] \geq \left(1 - \frac{2}{n}\right) \left(\frac{1}{2} - \frac{2}{\pi\gamma} \sqrt{\frac{\log n}{n}} - \frac{1}{n}\right) \geq \frac{1}{2} - \frac{2}{\pi\gamma} \sqrt{\frac{\log n}{n}} - \frac{3}{n}.
\]
Substituting this in Lemma~\ref{lem:min_shapley_bound}, we obtain
\begin{align*}
\E[\sv_1] &\geq \frac{1}{2(n-\ell)} - \frac{\ell}{n(n-\ell)} + \frac{1}{n-\ell} \left(\frac{1}{2} - \frac{2}{\pi\gamma} \sqrt{\frac{\log n}{n}} - \frac{3}{n}\right) - O\left(\frac{1}{n^2}\right) \\ 
					&=  \frac{1}{n-\ell} - \frac{\ell}{n(n-\ell)} - \frac{1}{n-\ell} O_\gamma \left(\sqrt{\frac{\log n}{n}}\right) - O\left(\frac{1}{n^2}\right) =
\frac{1}{n} - O_\gamma \left(\sqrt{\frac{\log n}{n^3}}\right).
\end{align*}
\qed\end{proof}


\section{The Balls and Bins Distribution: the Exponential Case}
\label{sec:balls_bins_exponential}
In Section~\ref{sec:bbuniform}, we show that even when the distribution is not inherently biased towards any agent, 
substantial inequalities may arise due to random noise.
We now turn to study the case in which the distribution is strongly biased. 
Returning to our formal definition of the general balls and bins process, we assume that the probabilities in the vector $\vec p$ are ordered in increasing order and $\frac{p_i}{p_{i+1}} = \rho$, for some $\rho < 1/2$.
We observe that as $m$ approaches $\infty$, the weight vector follows a power law with probability $1$, where for each $i=1,\ldots,n-1$, $\frac{w_i}{w_{i+1}} = \rho$.
A closely related family of weight vectors that we will refer to is the family of \emph{super-increasing} weight vectors:
\begin{definition}
	A series of positive weights $\vec w=(w_1,\ldots,w_n)$ is said to be super-increasing (SI) if for every $i=1,\ldots,n$, $\sum_{j=1}^{i-1}w_j < w_i$.
\end{definition}

The following three results (Lemma~\ref{lem:super-increasing}, Lemma~\ref{lem:super-increasing2} and Theorem~\ref{thm:super-increasing3}) show that for a sufficiently large value of $m$, estimating the Shapley values in WVGs where the weights are sampled from an exponential distribution can be reduced to the study of Shapley values in a game with a prescribed (fixed) SI weight vector; Section~\ref{sec:super-increasing} studies power distribution in WVGs with SI weights.
\jo{I changed following two theorems to lemmas, also, after verifying the it all works, we will move their proofs to the appendix.}
The following lemma gives a characterization of the necessary size of the voter population, so as to make the weight vector super-increasing, if the voters vote according to the above exponential distribution.
\begin{restatable}{lem}{superincreasing}
  \label{lem:super-increasing}
Assume that $m$ voters submit the votes according to the exponential distribution over candidates, such that for $\rho \in (0,\frac{1}{2})$, and the probability that voter $j$ votes for candidate $i$ is proportional to $\rho^{n-i}$. There is a constant $C > 0$ such that if $m \geq C\rho^{-n} (2\rho-1)^{-2} \log n$ then the resulting weight vector is super-increasing with probability $1 - O(\frac{1}{n})$. Furthermore, as $m\to\infty$, the probability approaches~$1$.

\yf{Replaced the bound on $m$ with something nicer but equivalent (I lost a factor of $\rho$ to make it look even nicer).}
\end{restatable}
The proof of the lemma uses a standard concentration bound (see Appendix~\ref{app:balls-and-bins-exponential}).

Before we proceed, it would be helpful to provide some intuition about the behavior of the Shapley values. Assuming that agent weights
are given by an $n$-length (increasing) sequence $\vec w$ of real-values, consider the set of all distinct subset sums of the weights $\mathcal{S}(\vec{w}) = \{ s : \exists P \subseteq [n] \text{ s.t. } s=\sum_{i \in P}w_{n+1-i}\}$ (we use $w_{n+1-i}$ instead of $w_i$ to make some formulas below nicer). Furthermore, suppose that the subset sums are ordered in increasing order; i.e., $\mathcal{S}(\vec{w})=\{ s_j\}_{j=1}^t$, such that $s_j < s_{j+1}$ for $1 \leq j < t$. It is easy to show, using the definition of the Shapley value, that for any quota $q \in (s_j,s_{j+1}]$, for $1 \leq j < t$, the Shapley values of every agent $i \in N$ remain constant at some value $\sv_i(j)$, defined for the $j$'th interval. We formalize this intuition in Section~\ref{sec:super-increasing}, where we give a formula for $\sv_i(j)$.

Before we state the formula (Proposition~\ref{pro:super-increasing-formula} below), we need some notation. For each $P \subseteq N$, let $\tilde w(P) = \sum_{i \in P} w_{n+1-i}$. For some $j$, $\tilde w(P) = s_j$, where $s_j \in \mathcal{S}$. If $P \neq N$ then $j < t$ and so $s_{j+1} = \tilde w(P^+)$ for some $P^+ \subseteq N$. Write $I^{\vec w}_P = (\tilde w(P), \tilde w(P^+)]$. Then by definition, the intervals $I^{w}_P$ partition the interval $(0,w(N)]$. We can now state the formula for $\sv_i(j)$. Given a weight vector $\vec w$, let $\sv^{\vec w}_i(q)$ denote the Shapley value of player $i$ when the quota is $q$ and the weights are given by $\vec w$.

\begin{proposition} \label{pro:super-increasing-formula}
 Suppose that $\vec w=(w_1,\ldots,w_n)$ is a SI sequence of weights, and suppose that $q \in (0,w(N)]$, say $q \in I^{\vec w}_P$ for some $P \subseteq N$. Write $P = \{j_0,\ldots,j_r\}$ in increasing order. If $i \notin P$ then
$\shapley^{w}_{n+1-i}(q) = \sum_{\substack{t \in \{0,\ldots,r\}\colon\\ j_t > i}} \frac{1}{j_t\binom{j_t-1}{t}}$. 
 If $i \in P$, say $i = j_s$, then
$\shapley^{\vec w}_{n+1-i}(q) = \frac{1}{j_s\binom{j_s-1}{s}} - \sum_{\substack{t \in \{0,\ldots,r\}\colon\\ j_t > i}} \frac{1}{j_t\binom{j_t-1}{t-1}}$.
\end{proposition}
Suppose that $\vec w$ is generated using a Balls and Bins process with probabilities $\vec p$, where $\vec p$ is a SI sequence; then it stands to reason that if a sufficiently large number of balls is tossed (i.e., $m$ is large enough), then voting power distribution under $\vec w$ will be very close to power distribution under the weight vector $\vec p$. This intuition is captured in the following lemma, which is proved in the appendix.

\begin{restatable}{lem}{superincreasingtwo}
\label{lem:super-increasing2}
 Suppose that $\vec{p}=(p_1,\ldots,p_n)$ is a SI sequence summing to $1$, and let $w_1,\ldots,w_n$ be obtained by sampling $m$ times from the distribution $p_1,\ldots,p_n$.

 Suppose that $T \in (0,1]$, say $T \in I^{\vec{p}}(P)$ for some $P \subseteq \{1,\ldots,n\}$. If the distance of $T$ from the endpoints $\tilde p(P),\tilde p(P^+)$ of $I^{\vec{p}}(P)$ is at least $\Delta = \sqrt{\log (nm)/m}$ then with probability $1 - \frac{2}{(nm)^2}$ it holds that if $\vec{w}$ is SI then for all $i \in N$, $\shapley_i^{\vec{w}}(mT) = \shapley_i^{\vec{p}}(T)$.
\end{restatable}

Combining both lemmas, we obtain our main result on the exponential case of the Balls and Bins distribution.

\begin{theorem} \label{thm:super-increasing3}
Assume that $m$ voters submit the votes according to the exponential distribution over candidates, such that for $\rho \in (0,\frac{1}{2})$, 
and the probability that voter $j$ votes for candidate $i$ is proportional to $\rho^{n-i}$. Assume further that $m \geq C\rho^{-n} (2\rho-1)^{-2} \log n$, where $C>0$ is some global constant.

 Suppose that $T \in (0,1]$, say $T \in I^{\vec{p}}(P)$ for some $P \subseteq \{1,\ldots,n\}$. If the distance of $T$ from the endpoints $\tilde p(P),\tilde p(P^+)$ of $I^{\vec{p}}(P)$ is at least $\Delta = \sqrt{\log (nm)/m}$ then with probability $1 - O(1/n)$ it holds that for all $i \in \{1,\ldots,n\}$, $\shapley_i^{\vec{w}}(mT) = \shapley_i^{\vec{p}}(T)$.

 Furthermore, for all but finitely many values of $T \in (0,1]$, the probability that $\shapley_i^{\vec w}(mT) = \shapley_i^{\vec{p}}(T)$ tends to $1$ as $m\to\infty$.
\end{theorem}
\begin{proof}
 Lemma~\ref{lem:super-increasing} gives a constant $C > 0$ such that if $m \geq C\rho^{-n} (2\rho-1)^{-2} \log n$ then $\vec{w}$ is SI with probability $1 - O(1/n)$. Hence the first part of the theorem follows from Lemma~\ref{lem:super-increasing2}.

 For the second part, Lemma~\ref{lem:super-increasing} shows that as $m\to\infty$, the probability that $\vec w$ is SI approaches $1$. Suppose now that $T$ is \emph{not} of the form $\tilde p(P)$ (these are the finitely many exceptions). When $m$ is large enough, the conditions of Lemma~\ref{lem:super-increasing2} are satisfied, and so as $m\to\infty$, the error probability in that lemma goes to $0$. The second part of the theorem follows.
\qed\end{proof}

The theorem shows that in the case of the exponential distribution, if the number of balls is large enough then we can calculate with high probability the Shapley values of the resulting distribution based on the Shapley values of the original exponential distribution (without sampling). It therefore behooves us to study the Shapley values of an exponential distribution, or indeed any SI sequence.

\section{Super-increasing sequences} \label{sec:super-increasing}
\yz{Removed road map from here, we have Section~\ref{sec:overview}}
In Section~\ref{sec:balls_bins_exponential}, we have shown that when the number of voters is large, studying the distribution of Shapley values when weights are drawn from an exponential balls and bins distribution boils down
to the study of the Shapley values where weights are super-increasing.
This section constitutes a thorough analysis of power distribution when weights are super-increasing; in particular, we provide strong generalizations of the results by \cite{zick2011sv} and \cite{zuck12manip}.


%
%

Up to this point, we assumed that the weights are arranged in non-decreasing order.
In order to simplify our formulas, we will somewhat abuse
our definitions by assuming that the weights are rather ordered in
non-increasing order, $w_1 > w_2 > \cdots > w_n > 0$.
We also assume that $\vec w$ is a super-increasing
sequence; that is, a sequence satisfying
$w_i > \sum_{j=i+1}^n w_j$
for all $i \in N$.

When considering different weight vectors, we will use $\sv_i^{\vec{w}}(q)$ for the Shapley value of agent~$i$ under weight vector~$\vec{w}$ and quota~$q$.
%
\yz{I have written a brief road map above and omitted this part, hope it is sufficient}
\subsection{Reducing super-increasing weight vectors to the case of a
  power law of $2$}
While not every quota in the range $(0,w(N)]$ can be expanded as a sum of members of $\{w_1,\ldots,w_n\}$, there are certain naturally defined intervals that partition $(0,w(N)]$. For a subset $C \subseteq N$, define $\beta(C) = \sum_{i \in C} 2^{n-i}$.
Intuitively, we think of $\beta(C)$ as the value resulting from the binary
characteristic vector of the set of agents $C$. The purpose of the
following two lemmas is to reduce every super-increasing weight
vector to the case where the weights obey a power-law distribution,
with a power of $2$.
\begin{lemma} \label{lem:super-increasing-rep}
 Let $C_1,C_2 \subseteq N$. Then $\beta(C_1) < \beta(C_2)$ if and only if $w(C_1) < w(C_2)$.
\end{lemma}
\begin{proof}
 In order to prove the claim, it suffices to observe adjacent sets $C_1,C_2\subseteq N$, i.e., ones satisfying $\beta(C_2) = \beta(C_1)+1$. Let $\ell = \max (N \setminus C_1)$, and define $C = C_1 \cap \{1,\ldots,\ell-1\}$. Then $C_1 = C \cup \{\ell+1,\ldots,n\}$ and $C_2 = C \cup \{\ell\}$.
 Therefore $w(C_2) - w(C_1) = w_\ell - w(\{\ell+1,\ldots,n\}) > 0$, since $w_1,\ldots,w_n$ is super-increasing.
\qed\end{proof}


For a non-empty set of agents $C \subseteq N$, we let
$P^- \subseteq N$ be the unique subset of agents
satisfying $\beta(P^-) = \beta(P) - 1$. Lemma~\ref{lem:super-increasing-rep}
shows that every quota $q \in (0,w(N)]$ belongs to a unique interval $(w(P^-),w(P)]$; we denote $P$ by $\pseqall(q)$.
We think of $\pseqall(q)$ as an increasing sequence $\pseq_0,\ldots,\pseq_r$ depending on $q$, for some value of $r$ which also depends on $q$. Whenever we write $P = \{a_0,\ldots,a_r\}$, we will always assume that $a_0 < \dots < a_r$.

\begin{lemma} \label{lem:super-increasing-P}
For all agents $i \in N$ and quotas $q \in (0,w(N)]$, $\sv_i^{\vec w}(q) = \sv_i^{\vec b}(\beta(\pseqall(q)))$, where $\vec b = (2^{n-1},\ldots,1)$.
\end{lemma}
\begin{proof}
Let $\sigma$ be a random permutation in $\perms{n}$, and recall that $P_i(\sigma)$ is the set of agents appearing before agent~$i$ in~$\sigma$.
The Shapley value $\sv_i^{\vec w}(q)$ is the probability that $w(P_i(\sigma)) \in [q-w_i,q)$, or equivalently, that $q \in (w(P_i(\sigma)),w(P_i(\sigma))+w_i]$. Since the intervals $(w(C^-),w(C)]$ partition $(0,w(N)]$, $q$ is in $(w(P_i(\sigma)),w(P_i(\sigma))+w_i]$ if and only if
$w(P_i(\sigma)) \leq w(\pseqall(q)^-)$ and $w(\pseqall(q)) \leq w(P_i(\sigma) \cup \{i\})$.
Lemma~\ref{lem:super-increasing-rep} shows that this is equivalent to checking whether $\beta(P_i(\sigma)) \leq \beta(\pseqall(q)^-)$ and $\beta(\pseqall(q)) \leq \beta(P_i(\sigma) \cup \{i\})$. Now, note that $\beta(\pseqall(q)^-) = \beta(\pseqall(q))-1$, so the above condition simply states that $i$ is pivotal for $\sigma$ under $\vec b$ when the quota is $\beta(\pseqall(q))$.
\qed\end{proof}
Lemma~\ref{lem:super-increasing-P} implies that for any super-increasing $\vec w$, if we wish to compute $\sv_i^{\vec w}(q)$, it is only necessary to find $\pseqall(q)$. However, finding $\pseqall(q)$ is easy; a greedy algorithm can find $\pseqall(q)$ in linear time (see Appendix~\ref{sec:si-greedy}). In the special case in which $w_i = d^{n-i}$ for some integer $d$, there is a particularly simple formula described in Appendix~\ref{sec:si-dary}.


We now present a closed-form formula for the Shapley values, whose proof is given in the appendix. The resulting Shapley values are illustrated in Figure~\ref{fig:super-increasing}.

\begin{restatable}{thm}{superincreasingformula} \label{thm:super-increasing-formula}
 Consider an agent $i \in N$ and a prescribed quota value $q\in (0,w(N)]$. Let $\pseqall(q) = \{\pseq_0,\ldots,\pseq_r\}$.
 If $i \notin \pseqall(q)$ then
$\shapley_i(q) = \sum_{\substack{t \in \{0,\ldots,r\}\colon\\ \pseq_t > i}} \frac{1}{\pseq_t\binom{\pseq_t-1}{t}}$.
 If $i \in \pseqall(q)$, say $i = \pseq_s$, then
$\shapley_i(q) = \frac{1}{\pseq_s\binom{\pseq_s-1}{s}} - \sum_{\substack{t \in \{0,\ldots,r\}\colon\\ \pseq_t > i}} \frac{1}{\pseq_t\binom{\pseq_t-1}{t-1}}$.
\end{restatable}

\begin{figure}[ht!]
 \centering
 \begin{subfigure}[t]{0.4\textwidth}
  \centering \includegraphics[width=\textwidth]{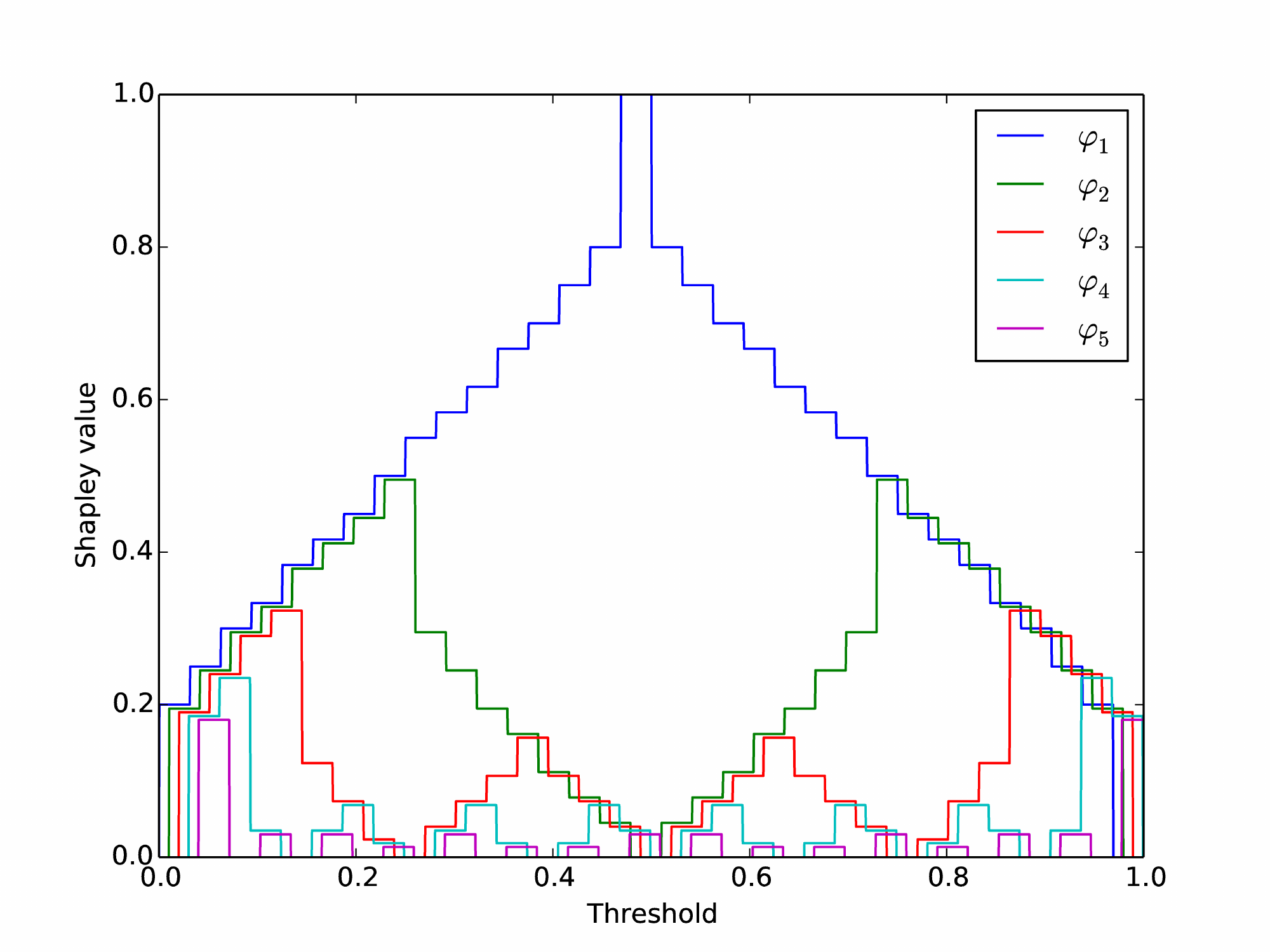}
  \caption{Shapley values for $n = 5$, $w_i = 2^{-i}$. Values $\shapley_i(q)$ for different $i$ are slightly nudged to show the effects of Lemma~\ref{lem:super-increasing-adjacent}.}
 \end{subfigure} \hspace{0.05\textwidth}
 \begin{subfigure}[t]{0.4\textwidth}
  \centering \includegraphics[width=\textwidth]{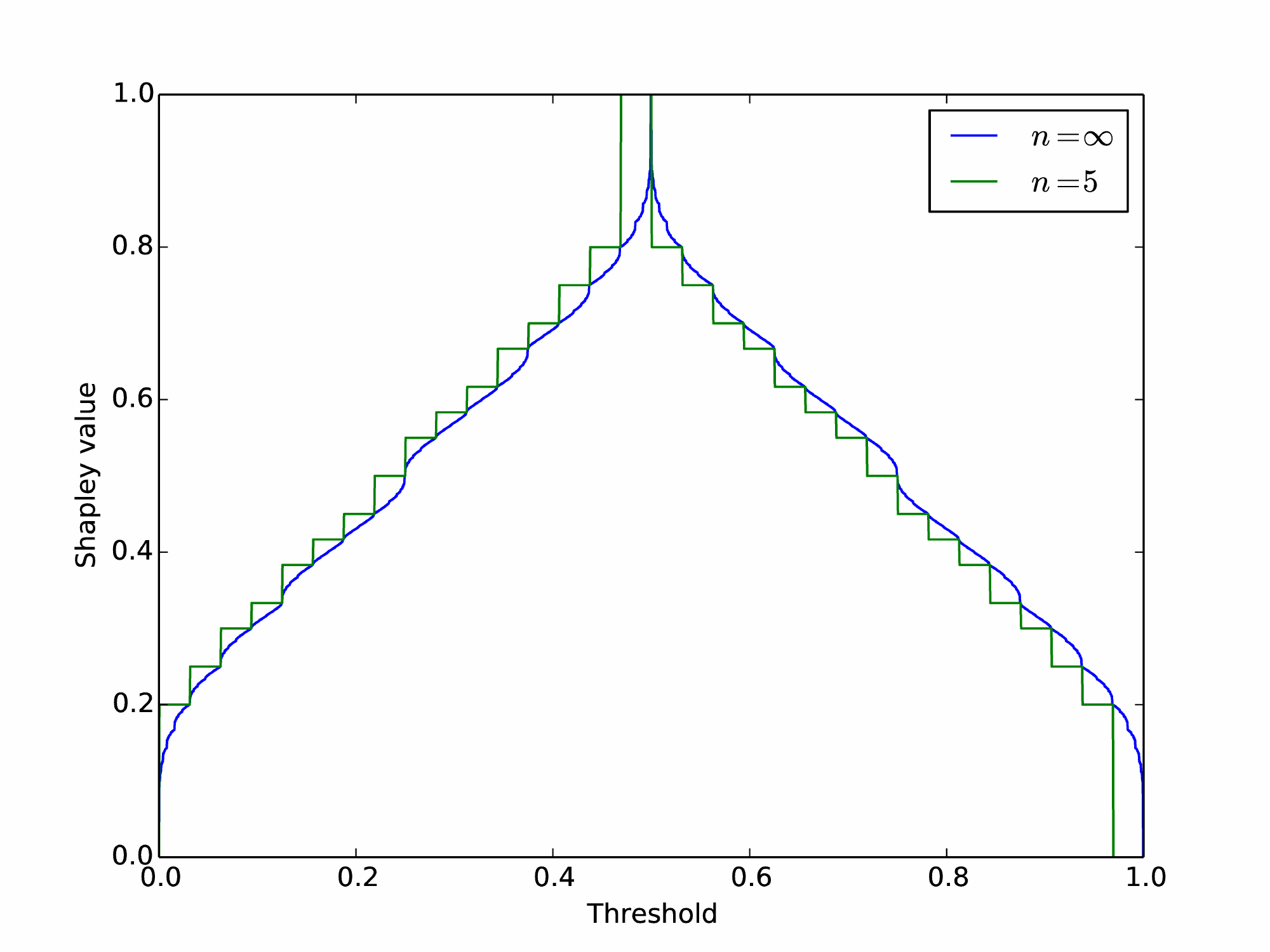}
  \caption{Shapley values $\shapley_1(q)$ for $n = 5$, $w_i = 2^{-i}$ compared to the limiting case $n = \infty$.}
 \end{subfigure}

 \begin{subfigure}[t]{0.4\textwidth}
  \centering \includegraphics[width=\textwidth]{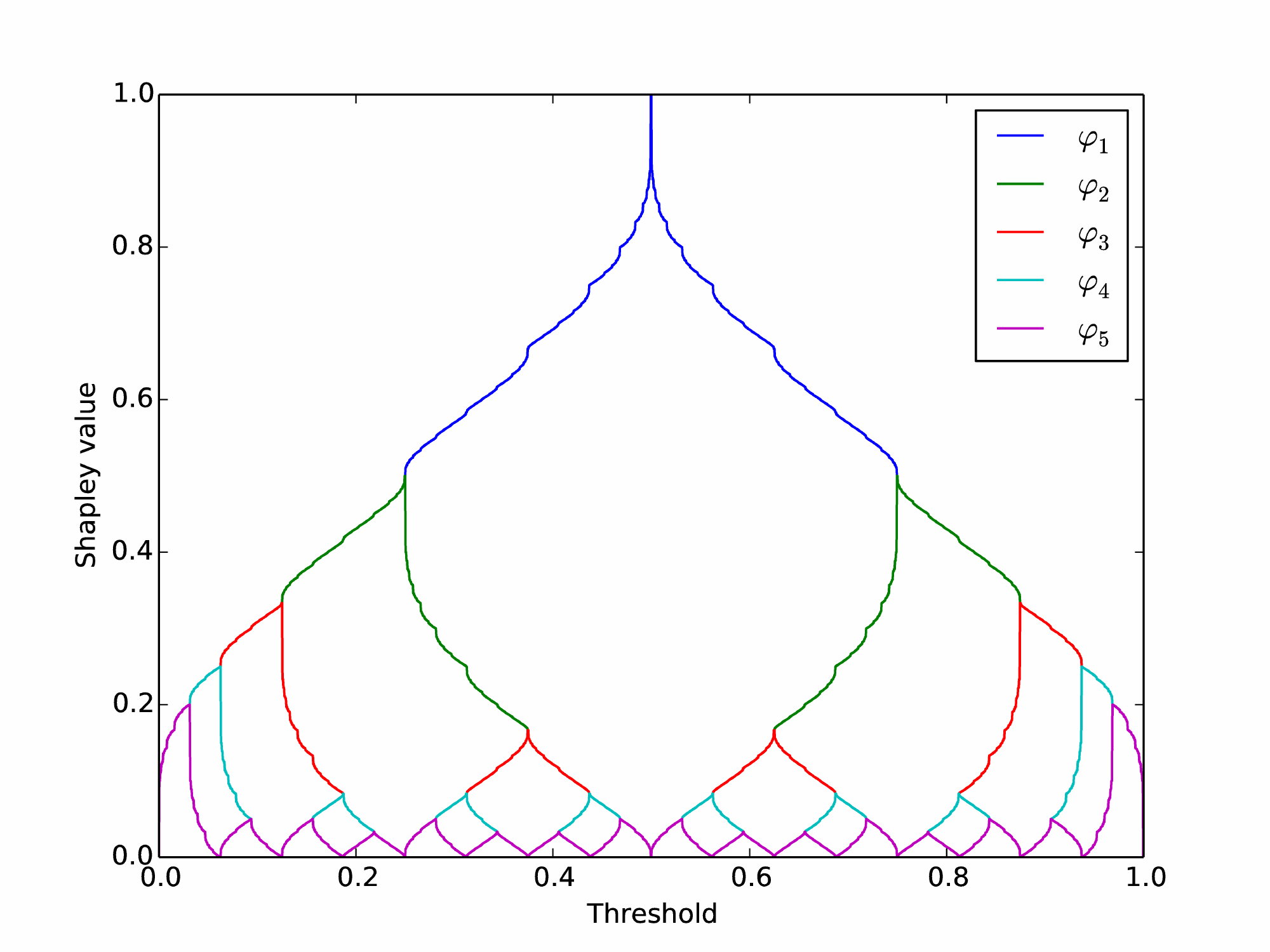}
  \caption{Shapley values in the limiting case, $w_i = 2^{-i}$.}
 \end{subfigure} \hspace{0.05\textwidth}
 \begin{subfigure}[t]{0.4\textwidth}
  \centering \includegraphics[width=\textwidth]{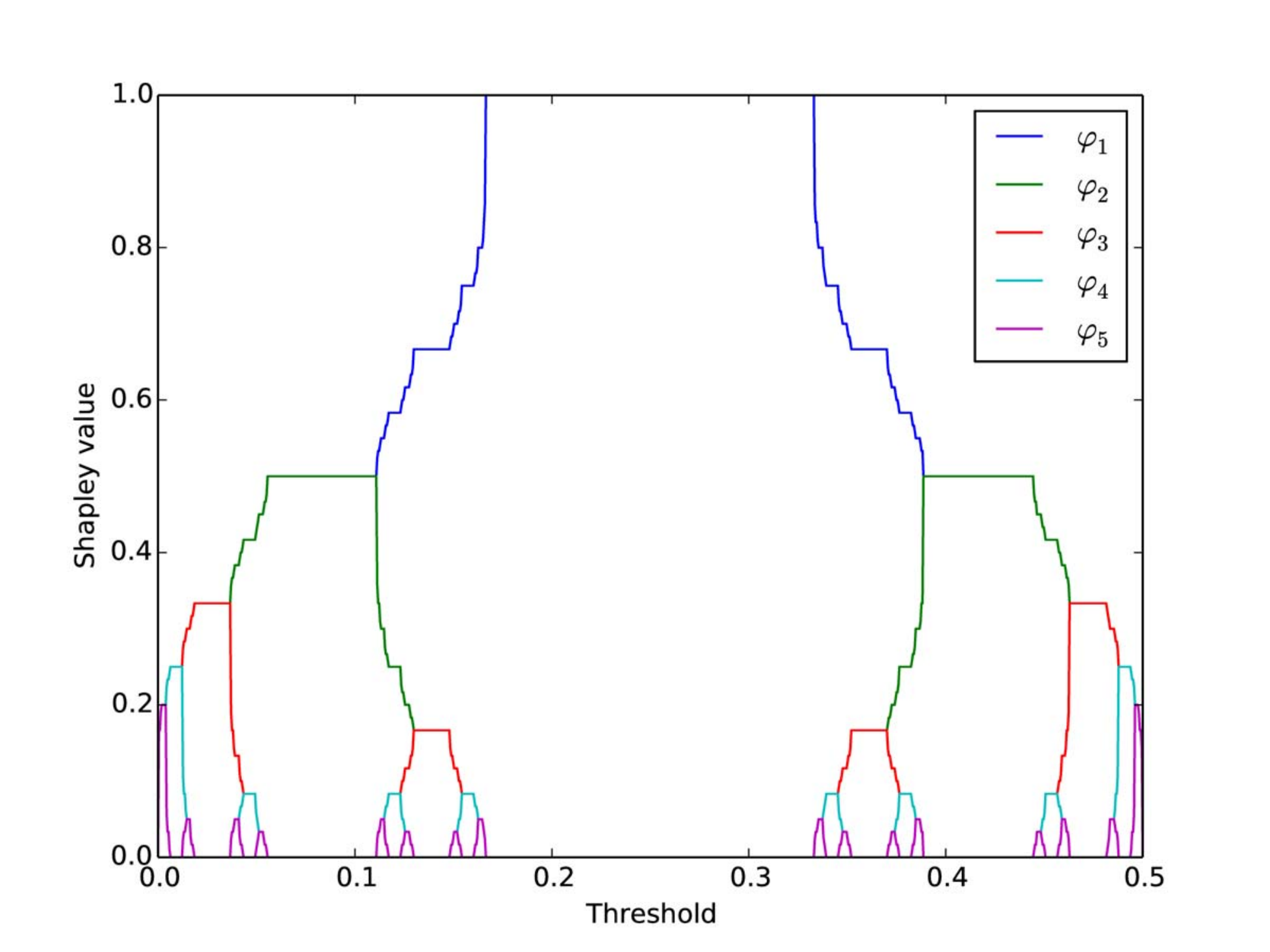}
  \caption{Shapley values in the limiting case, $w_i = 3^{-i}$.}
 \end{subfigure}
 \caption{Examples of several Shapley values corresponding to super-increasing sequences.}
 \label{fig:super-increasing}
\end{figure}

\subsection{Properties of the Shapley values}
\label{sec:Super-increasing-props}

Zuckerman et al.~\cite{zuck12manip} provide a nice characterization of super-increasing sets:
\begin{theorem}[\cite{zuck12manip}]
\label{thm:zuckerman}
If the weights $\vec w$ are super-increasing then for every quota $q \in (0,w(N)]$, either $\sv_1(q) = \sv_2(q)$ or $\sv_2(q) = \sv_3(q)$.
\end{theorem}

In this section, we further generalize this result, using Theorem~\ref{thm:super-increasing-formula}.
Specifically, as a consequence of the theorem, we can determine in
which cases $\shapley_i(q) = \shapley_{i+1}(q)$. The results are summarized in the following lemma, which is proved in the appendix.
Given a set $S \subseteq N$, let $\Psi_i(S)$ be the indicator variable of $i \in S$; that is, $\Psi_i(S) = 1$ if $i \in S$, and is 0 otherwise.
\begin{restatable}{lem}{superincreasingadjacent} \label{lem:super-increasing-adjacent}
 Given a quota $q \in (0,w(N)]$, let $\pseqall(q) = \{\pseq_0,\ldots,\pseq_r\}$. For each $i \in N\setminus\{n\}$, if $\Psi_i(\pseqall(q)) = \Psi_{i+1}(\pseqall(q))$ then $\sv_i(q) = \sv_{i+1}(q)$. If $\Psi_i(\pseqall(q)) = 0$ and $\Psi_{i+1}(\pseqall(q)) =1$ then $\shapley_i(q) \geq \shapley_{i+1}(q)$, with equality if and only if $i+1 = \pseq_r$.
 If $\Psi_i(\pseqall(q)) = 1$ and $\Psi_{i+1}(\pseqall(q)) = 0$ then $\shapley_i(q) > \shapley_{i+1}(q)$.
\end{restatable}

For each $i \in N$, let $\Psi_i$ be the truth value of $i \in \pseqall(q)$. Lemma~\ref{lem:super-increasing-adjacent} shows that if $\Psi_i = \Psi_{i+1}$ then $\shapley_i(q) = \shapley_{i+1}(q)$. Since there are only two possible truth values, for each $i \in N \setminus \{1,n\}$, either $\shapley_{i-1}(q) = \shapley_i(q)$ or $\shapley_i(q) = \shapley_{i+1}(q)$. This generalizes Theorem~\ref{thm:zuckerman}.

Since the Shapley values are constant in the interval $(w(P^-),w(P)]$, it follows that in order to analyze the behavior of $\sv_i(q)$, one need only determine the rate of increase or decrease at quotas of the form $w(P)$ for $P \subseteq N$. These are given by the following lemma, proved in the appendix.
\begin{restatable}{lem}{superincreasingjumps} \label{lem:super-increasing-jumps}
 Let $P \subseteq N$ be a non-empty set of agents, and let $i \in N$ be an agent. If $i \notin P^-$ then $\shapley_i(w(P^-)) < \shapley_i(w(P))$. If $i \in P^-$ then $\shapley_i(w(P^-)) > \shapley_i(w(P))$.

 Moreover, $|\shapley_i(w(P)) - \shapley_i(w(P^-))| \leq \frac1n$; this inequality is tight only if
\begin{inparaenum}[\itshape a\upshape)]
\item $P = \{n\}$.
\item $i < n$ and $P = \{1,\ldots,i\}$ or $P = \{i,n\}$, or
\item $i = n$ and $P = \{n-1\}$.
\end{inparaenum}
Otherwise, $|\shapley_i(w(P)) - \shapley_i(w(P^-))| \leq \frac{1}{n(n-1)}$.
\end{restatable}

\subsection{Limiting case}
In what follows, we briefly discuss some interesting properties of the Shapley value for weight vectors that are finite prefixes of an infinite super-increasing sequence. The full details can be found in Appendix~\ref{app:super-increasing-limits}.
Section~\ref{sec:Super-increasing-props} shows that the Shapley value be easily expressed and analyzed when weights are super-increasing. It is in fact useful to think of classes of super-increasing weights that are finite prefixes of an infinite super-increasing sequence; one example of such a sequence is $\left(2^{-i}\right)_{i = 0}^\infty$. Given an infinite super-increasing weight sequence $\vec w$, we refer to the first $m$ elements of $\vec w$ as $\vec w|_m$. Given an infinite super-increasing sequence $\vec w$, we can use the closed-form formula we define for the Shapley value in the finite case in order to derive a value for the infinite case. We can then define a value for agent $i$ under an infinite sequence $\vec w$: $\sv_i^{\vec w}(q)$. We show that the two formulas are closely related; one can derive $\sv_i^{\vec w|_m}(q)$ using $\sv_i^{\vec w}(q)$ (the connection is illustrated in Figure~\ref{fig:super-increasing}); moreover, we show that $\sv_i^{\vec w}(q)$ is continuous in $q$, and that it takes values of 0 when $q$ approaches 0, and when $q$ approaches $\sum_{i = 0}^\infty w_i$.


\section{Conclusions and Future Work}
We have studied the Shapley value as a function of the quota under
a number of natural weight distributions. Assuming that weights are drawn from balls and bins distributions allows us to reason rigourously about the effect of quota changes. We were also able to completely characterize the case where weights are super-increasing, strongly generalizing previous work.
The take-home message from our work is that changes to the quota matter, even when weights are nearly identical. 
Given the relative success of this analysis, it would be interesting to study other natural weight distributions (the case of i.i.d. bounded weights is studied in an extended version of this paper~\cite{oren2014sv}). Moreover, our results show that employing probabilistic approaches to cooperative games (beyond the case of WVGs) may be a useful research avenue.
\bibliographystyle{splncs}
\bibliography{bib}
\newpage

\appendix{Appendix}

\section{Missing Proofs from Section~\ref{sec:bbuniform}}
\subsection{Proof of Lemma~\ref{lem:binomial-formula}}
We now provide the complete proof of Lemma~\ref{lem:binomial-formula}
\binomialformula*
\begin{proof}
 Assume without loss of generality that $w_j \geq w_i$, and so $\sv_j \geq \sv_i$.
 For $\sigma \in \perms{n}$, let $T_{ij}(\sigma)$ be the permutation obtained by exchanging agents $i$ and $j$. Then by the definition of the Shapley value and by linearity of expectations:
\begin{align*}
\sv_j - \sv_i &= \Ex_{\sigma \in \perms{n}} (m_j(\sigma) - m_i(\sigma)) \\
&= \Ex_{\sigma \in \perms{n}} m_j(\sigma) - \Ex_{\sigma \in \perms{n}} m_i(\sigma) = \Ex_{\sigma \in \perms{n}} (m_j(T_{ij}(\sigma)) - m_i(\sigma)).
\end{align*}
We proceed to evaluate $m_j(T_{ij}(\sigma)) - m_i(\sigma)$. Suppose first that agent $i$ precedes agent $j$ in $\sigma$, so that $\sigma = S \; i \; R \; j \; U$ and $T_{ij}(\sigma) = S \; j \; R \; i \; U$ (where $S,R$, and $U$ form a partition of $N \setminus \{i,j\}$). In this case $m_j(T_{ij}(\sigma)) - m_i(\sigma) \neq 0$ precisely when $w(S) + w_i < q \leq w(S) + w_j$, in which case $m_j(T_{ij}(\sigma)) - m_i(\sigma) = 1$; we can rewrite the condition as $w(S) \in [q - w_j, q-w_i)$.

When agent $j$ precedes agent $i$ in $\sigma$, we can write $\sigma = S \; j \; R \; i \; U$ and $T_{ij}(\sigma) = S \; i \; R \; j \; U$. In this case $m_j(T_{ij}(\sigma)) - m_i(\sigma) \neq 0$ precisely when $w(S) + w_i + w(R) < q \leq w(S) + w_j + w(R)$, in which case $m_j(T_{ij}(\sigma)) - m_i(\sigma) = 1$; we can rewrite the condition as $w(S \cup R) \in [q - w_j, q - w_i)$.

 In order to unify both conditions together, define $P'_i(\sigma) = P_i(\sigma) \setminus \{j\}$. Using this definition, we see that $m_j(T_{ij}(\sigma)) - m_i(\sigma)$ is the indicator of the event $w(P'_i(\sigma)) \in [q - w_j, q - w_i)$. The cardinality $|P'_i(\sigma)|$ is exactly the position of agent $i$ in the permutation $\sigma'$ obtained by removing agent $j$ from $\sigma$, minus one. Since $\sigma$ is a uniformly random permutation of $N$, $\sigma'$ is a uniformly random permutation of $N \setminus \{j\}$, and so $|P'_i(\sigma)|$ is distributed randomly among $\{0,\ldots,n-2\}$. Given $|P'_i(\sigma)|$, the set $P_i(\sigma)$ is chosen randomly among all subsets of $N \setminus \{i,j\}$ of the specified size, yielding our formula.
\end{proof}
\subsection{Proof of Theorem~\ref{thm:plurality_positive}}
We give the full proof of the following theorem.

\pluralitypositive*

In this section, we do not assume that the weights $w_1,\ldots,w_n$ are ordered, in order to maintain the fact that the weights are independent random variables.

The idea of the proof is to use the following criterion, which is a
consequence of Lemma~\ref{lem:binomial-formula}:
\begin{proposition} \label{pro:equality-criterion}
Suppose that for all agents $i,j \in N$ and for all subsets $S \subseteq N \setminus \{i,j\}$, we have $q \notin (w(S \cup \{i\}), w(S \cup \{j\})]$.
Then all Shapley values are equal to $1/n$.
\end{proposition}
\begin{proof}
We show that under the assumption on $q$, all Shapley values are equal, and so all must equal $1/n$. Suppose that for some agents $i \neq j$, we have $\sv_i < \sv_j$ (and so $w_i < w_j$). Lemma~\ref{lem:binomial-formula} implies the existence of a set $S \subseteq N \setminus \{i,j\}$ satisfying $q - w_j \leq w(S) < q - w_i$, or in other words $w(S) + w_i < q \leq w(S) + w_j$. This is exactly what is ruled out by the assumption on $q$.
\end{proof}

Next, we show that the weights $w(S)$ are concentrated around points of the form $\ell\frac{m}{n}$.
\begin{lemma} \label{lem:individual-concentration}
Suppose that $m > 3n^2$. With probability $1 - 2(\frac{2}{e})^n$, the following holds: for all $S \subseteq N$,
$|w(S) - \frac{|S|m}{n}| \leq \sqrt{3nm}$.
\end{lemma}
\begin{proof}
The proof uses a straightforward Chernoff bound. We can assume that $S
\neq \emptyset$ (as otherwise the bound is trivial). For each non-empty set $S \subseteq N$, the distribution of $w(S)$ is $\bin(m,\frac{|S|}{n})$.
Therefore for $0 < \delta < 1$,
\[ \pr\left[ \left|w(S) - \frac{|S|m}{n}\right| > \delta \frac{|S|m}{n} \right] \leq 2e^{-\frac{\delta^2|S|m}{3n}}. \]
Choosing $\delta = \sqrt{\frac{3n^2}{|S|m}} < 1$, we obtain
\[ \pr\left[ \left|w(S) - \frac{|S|m}{n}\right| > \sqrt{3|S|m} \right] \leq 2e^{-n}. \]
Since there are $2^n$ possible sets $S$, a union bound implies that $|w(S) - \frac{|S|m}{n}| \leq \sqrt{3nm}$ with probability at least $1 - 2(\frac{2}{e})^n$.
\end{proof}

This immediately implies Theorem~\ref{thm:plurality_positive}, as we now show.
\begin{proof}[of Theorem~\ref{thm:plurality_positive}]
First, note that  $M < 1$, as otherwise, it would imply that for all
$\ell =1,\ldots,n$, $| q - \ell m/n| \geq m/n$, which is impossible,
as every quota in the range $(0,m]$ is within some integral multiple
of $m/n$. Thus, having  $M > 1$, implies that $m > 3n^3 \geq 3n^2$, as
required by Lemma~\ref{lem:individual-concentration}.

Lemma~\ref{lem:individual-concentration} shows that with probability
$1 - 2(\frac{2}{e})^n$, for all sets $S$ we have $|w(S) -
\frac{|S|m}{n}| \leq \sqrt{3nm}$. Condition on this event. Suppose, for the sake of obtaining a contradiction, that $\sv_i < \sv_j$ for some agents $i,j$. Then Proposition~\ref{pro:equality-criterion} shows that there must exist some $S \subseteq N \setminus \{i,j\}$ such that $q \in (w(S \cup \{i\}), w(S \cup \{j\})]$. Since both $w(S \cup \{i\})$ and $w(S \cup \{j\})$ are $\sqrt{3nm}$-close to $\frac{(|S|+1)m}{n}$, this implies that $|q - \frac{(|S|+1)m}{n}| \leq \sqrt{3nm} = \frac{1}{\sqrt{M}} \cdot \frac{m}{n}$, contradicting our assumption on $q$. We conclude that all agents have the same Shapley value $1/n$.
\end{proof}

\subsection{Proof of Lemma~\ref{lem:min_shapley_bound}}
We prove the following lemma.

\minshapleybound*

The proof closely follows the proof sketch in Section~\ref{sec:bbuniform-weak}.

We will need the fact that with high probability, $w_1$ is close to $m/n$.

\begin{lemma} \label{lem:min_weight_lb}
With probability at least $1 - 1/n$,
\[ \frac{m}{n} - \sqrt{\frac{4m\log n}{n}} \leq w_1 \leq \frac{m}{n}. \]
\end{lemma}
\begin{proof}
Clearly $w_1 \leq m/n$ always, so we only need to address the lower bound on $w_1$.
Let $w'_1,\ldots,w'_n$ be the loads of the bins before sorting them. The loads $w'_i$ are independent random variables with distribution $\bin(m,1/n)$. For each index $i$, Chernoff's bound shows that
\[
\pr\left[w'_i < \frac{m}{n} - \sqrt{\frac{4m\log n}{n}}\right] \leq e^{-\frac{4m\log n/n}{2m/n}} = \frac{1}{n^2}.
\]
A union bound shows that with probability $1 - 1/n$, all $i \in N$ satisfy $w'_i \geq \frac{m}{n} - \sqrt{\frac{4m\log n}{n}}$, and so $w_1 \geq \frac{m}{n} - \sqrt{\frac{4m\log n}{n}}$.
\end{proof}

Below we will be interested in bounding probabilities of the form $\pr_{A \inR \setbinom{N \setminus \{1\}}{k}}[P(w(A))]$ for predicates $P$. The following lemma shows how to bound these probabilities from above.

\begin{lemma} \label{lem:prob-est}
For a weight vector $\mathbf{w}$ and $S \subseteq N$, let
$\mathcal{E}(w(S))$ be a random event (i.e., some predicate on $w(S)$), and let $0 \leq k \leq n-1$. Then
\[
\pr_{A \inR \setbinom{N \setminus \{1\}}{k}} [\mathcal{E}(w(A))] \leq \frac{n}{n-k} \pr[\mathcal{E}(\bin(m,\tfrac{k}{n}))].
\]
Also,
\[
\pr_{A \inR \setbinom{N}{k}} [\mathcal{E}(w(A))] = \pr[\mathcal{E}(\bin(m,\tfrac{k}{n}))].
\]
\end{lemma}
\begin{proof}
First, we have
\begin{eqnarray*}
\pr_{A \inR \setbinom{N \setminus \{1\}}{k}} [\mathcal{E}(w(A))] & = &
\frac{1}{\binom{n-1}{k}} \sum_{A \in \setbinom{N \setminus \{1\}}{k}} \pr[\mathcal{E}(w(A))] \\
			& \leq &\frac{1}{\binom{n-1}{k}} \sum_{A \in \setbinom{N}{k}} \pr[\mathcal{E}(w(A))] \\
			& = &\frac{n}{n-k} \pr_{A \inR \setbinom{N}{k}} [\mathcal{E}(w(A))].
\end{eqnarray*}
Consider the last expression.
Since the probability is over all subsets of $N$ of size $k$, the same
value is obtained from the \emph{unsorted} Balls and Bins process (without sorting the loads). Under this process, $w(A) \sim \bin(m,\frac{k}{n})$ for all $A \in \setbinom{N}{k}$, and so
\[ \pr_{A \inR \setbinom{N}{k}} [\mathcal{E}(w(A))] = \pr_{w \sim \bin(m,\tfrac{k}{n})} [\mathcal{E}(w)]. \]
This implies the lemma.
\end{proof}

Let $p_k = \pr_{A \inR \setbinom{N \setminus \{1\}}{k}}[q - w_1 \leq w(A) < q]$, and recall that formula~\eqref{eq:shapley-alt-formula} shows that $\sv_1 = \frac{1}{n} \sum_{k=0}^{n-1} p_k$.
We start by showing that the only non-negligible $p_k$ are $p_{\ell-1}$ and $p_\ell$, using a Chernoff bound. The idea is that when $k \geq \ell+1$, it is highly unlikely that $w(A) < q$, and when $k \leq \ell -1$, it is highly unlikely that $w(A) \geq q - w_1$.

\begin{lemma} \label{lem:negligible-size}
 Suppose that $m \geq 9n^2\log n$. Then for $k \in \{1,\ldots,n\} \setminus \{\ell-1,\ell\}$ we have $p_k \leq 1/n^2$, and so
\[ 0 \leq \E[\sv_1] - \frac{p_{\ell-1} + p_{\ell}}{n}  \leq \frac{1}{n^2}. \]
\end{lemma}
\begin{proof}
 Let $k \in N$. Lemma~\ref{lem:prob-est} shows that
\[
 p_k \leq n \pr[q - w_1 \leq \bin(m,\tfrac{k}{n}) < q].
\]
 Suppose first that $k \geq \ell+1$. Chernoff's bound shows that
\[
 \pr[q - w_1 \leq \bin(m,\tfrac{k}{n}) < q] \leq \pr[\bin(m,\tfrac{k}{n}) < \tfrac{km}{n} - \tfrac{m}{n}] \leq e^{-\frac{(m/n)^2}{3km/n}} = e^{-m/(3nk)} \leq \frac{1}{n^3}.
\]
 Suppose next that $k \leq \ell-2$. Since $w_1 \leq m/n$, another application of Chernoff's bound gives 
\begin{eqnarray*}
 \pr[q - w_1 \leq \bin(m,\tfrac{k}{n}) < q] & \leq &\pr[\bin(m,\tfrac{k}{n}) \geq \tfrac{(\ell-1)m}{n}] \\
			& \leq &\pr[\bin(m,\tfrac{k}{n}) \geq \tfrac{km}{n} + \tfrac{m}{n}] \\
			& \leq & e^{-\frac{(m/n)^2}{3km/n}} = e^{-m/(3nk)} \leq \frac{1}{n^3}.
\end{eqnarray*}
 Therefore $p_k \leq 1/n^2$ for all $k \in N \setminus \{\ell-1,\ell\}$. The estimate for $\E[\sv_1]$ follows from formula~\eqref{eq:shapley-alt-formula}.
\end{proof}

The next step is to consider the following estimates for $p_{\ell-1},p_\ell$:
\begin{align*}
p'_{\ell-1} &= \pr_{A \inR \setbinom{N \setminus \{1\}}{\ell-1}}[q - w_1 \leq w(A)], \\
p'_\ell &= \pr_{A \inR \setbinom{N \setminus \{1\}}{\ell}}[w(A) < q].
\end{align*}
The following lemma shows that $p'_{\ell-1} \approx p_{\ell-1}$ and $p'_\ell \approx p_\ell$.

\begin{lemma} \label{lem:negligible-side}
 Suppose that $m \geq 24n^2\log n$. Then $p_{\ell-1} \leq p'_{\ell-1} \leq p_{\ell-1} + \frac{1}{n}$ and $p_\ell \leq p'_\ell \leq p_\ell + \frac{2}{n}$, and so
\[ -\frac{3}{n^2} \leq \E[\sv_1] - \frac{p'_{\ell-1} + p'_{\ell}}{n}  \leq \frac{1}{n^2}. \]
\end{lemma}
\begin{proof}
Clearly $p_{\ell-1} \leq p'_{\ell-1}$ and $p_{\ell} \leq p'_{\ell}$. First,
\[ p'_{\ell-1} - p_{\ell-1} \leq \pr_{A \inR \setbinom{N \setminus\{1\}}{\ell-1}} [w(A) \geq q] \leq n \pr[\bin(m,\tfrac{\ell-1}{n}) \geq q], \]
using Lemma~\ref{lem:prob-est}. Chernoff's bound shows that 
\[ \pr[\bin(m,\tfrac{\ell-1}{n}) \geq \tfrac{(\ell-1)m}{n} + \tfrac{m}{n}] \leq  e^{-\frac{(m/n)^2}{3(\ell-1)m/n}} = e^{-m/(3n(\ell-1))} \leq \frac{1}{n^2}. \]
Similarly,
\[ p'_\ell - p_\ell \leq \pr_{A \inR \setbinom{N \setminus\{1\}}{\ell}} [w(A) < q - w_1] \leq n \pr[\bin(m,\tfrac{\ell}{n}) < q - w_1]. \]
We now need the lower bound on $w_1$ given by Lemma~\ref{lem:min_weight_lb}, which holds with probability $1-1/n$:
\[ q - w_1 \leq \frac{\ell m}{n} - \left(\frac{m}{n} - \sqrt{\frac{4m\log n}{n}}\right) \leq \frac{\ell m}{n} - \frac{m}{2n}, \]
the latter inequality following from $m \geq 24n^2 \log n > 16n\log n$. Assuming the lower bound on $w_1$,
\[ \pr[\bin(m,\tfrac{\ell}{n}) < q - w_1] \leq e^{-\frac{(m/(2n))^2}{3(\ell-1)m/n}} = e^{-m/(12n(\ell-1))} \leq \frac{1}{n^2}. \]
Therefore
\[
p'_\ell - p_\ell \leq \left(1 - \frac{1}{n}\right) \cdot \frac{1}{n^2} + \frac{1}{n} \cdot 1 < \frac{2}{n}.
\]
The formula for $\E[\sv_1]$ follows from Lemma~\ref{lem:negligible-size}.
\end{proof}

It remains to relate $p'_{\ell-1}$ and $p'_\ell$.

\begin{lemma} \label{lem:relating-ell-ell-1}
 Suppose that $m \geq 24n^3\log n$. Then
\[ \left| p'_\ell - \left(\frac{n}{2(n-\ell)} - \frac{\ell}{n-\ell} (1 - p'_{\ell-1})\right)\right| \leq \frac{1}{n}, \]
and so
\[ -\frac{4}{n^2} \leq \E[\sv_1] - \left( \frac{1}{2(n-\ell)} - \frac{\ell}{n(n-\ell)} + \frac{p'_{\ell-1}}{n-\ell} \right)  \leq \frac{2}{n^2}. \]
\end{lemma}
\begin{proof}
We have
\begin{align*}
p'_\ell &= \pr_{A \inR \setbinom{N \setminus \{1\}}{\ell}}[w(A) < q] \\ &=
\frac{1}{\binom{n-1}{\ell}} \sum_{A \in \setbinom{N \setminus \{1\}}{\ell}} \pr[w(A) < q] \\ &=
\frac{1}{\binom{n-1}{\ell}} \sum_{A \in \setbinom{N}{\ell}} \pr[w(A) < q] - \frac{1}{\binom{n-1}{\ell}} \sum_{A \in \setbinom{N \setminus \{1\}}{\ell-1}} \pr[w(A) + w_1 < q] \\ &=
\frac{n}{n-\ell} \pr_{A \inR \setbinom{N}{\ell}} \pr[w(A) < q] - \frac{\ell}{n-\ell} \left(1 - \pr_{A \inR \setbinom{N \setminus \{1\}}{\ell-1}} [w(A) + w_1 \geq q]\right) \\ &=
\frac{n}{n-\ell} \pr[\bin(m,\frac{\ell}{n}) < q] - \frac{\ell}{n-\ell} (1 - p'_{\ell-1}),
\end{align*}
where the final equality follows from the second part of Lemma~\ref{lem:prob-est}. We proceed to estimate $\pr[\bin(m,\frac{\ell}{n}) < q]$ using the Berry--Esseen theorem. The normalized binomial $\bin(m,\frac{\ell}{n}) - q$ is a sum of $m$ independent copies of the random variable $X$ with $\pr[X = 1 - \frac{\ell}{n}] = \frac{\ell}{n}$ and $\pr[X = -\frac{\ell}{n}] = 1 - \frac{\ell}{n}$. The Berry--Esseen theorem states that
\[ |\pr[\bin(m,\frac{\ell}{n}) - q < 0] - \pr[\nor(0,\sigma^2) < 0]| < \frac{\rho}{\sigma^3 \sqrt{m}}, \]
where $\sigma^2 = \E[X^2] = \frac{\ell}{n} (1 - \frac{\ell}{n})^2 + (1-\frac{\ell}{n}) (\frac{\ell}{n})^2 = \frac{\ell}{n}(1-\frac{\ell}{n})$ and $\rho = \E[|X|^3] = \frac{\ell}{n} (1 - \frac{\ell}{n})^3 + (1-\frac{\ell}{n}) (\frac{\ell}{n})^3 = \frac{\ell}{n}(1-\frac{\ell}{n})[(\frac{\ell}{n})^2 + (1-\frac{\ell}{n})^2]$. Since $\pr[\nor(0,\sigma^2) < 0] = 1/2$, we conclude that
\[ \left|\pr[\bin(m,\frac{\ell}{n}) - q < 0] - \frac{1}{2}\right| < \frac{1}{\sqrt{m}} \frac{(\tfrac{\ell}{n})^2 + (1-\tfrac{\ell}{n})^2}{\sqrt{\tfrac{\ell}{n}\left(1-\tfrac{\ell}{n}\right)}} \leq 2\sqrt{\frac{n}{m}}, \]
since the denominator is at least $\sqrt{\frac{1}{n}(1-\frac{1}{n})}$, and the numerator is at most $2(1-\frac{1}{n})^2 \leq 2\sqrt{1-\frac{1}{n}}$.
Since $m \geq 24n^3\log n \geq 4n^3$, we further have $2\sqrt{\frac{n}{m}} \leq \frac{1}{n}$.

The formula for $\E[\sv_1]$ follows from Lemma~\ref{lem:negligible-side}.
\end{proof}

Lemma~\ref{lem:relating-ell-ell-1} is simply a reformulation of Lemma~\ref{lem:min_shapley_bound}.

\subsection{Proof of Lemma~\ref{lem:min_weight}}
Let us recall Lemma~\ref{lem:min_weight}.

\minweight*

We already proved the upper bound in Lemma~\ref{lem:min_weight_lb}, using a simple union bound. The lower bound (corresponding to an upper bound on $w_1$) is much more difficult, because of the dependence between the individual bins. One way to overcome this difficulty is to use the Poisson approximation, given by the following theorem.

\begin{theorem}[\cite{upfalProbBook}]
    \label{thm:poisson_approx}
 Let $w_1,\dots,w_n$ be sampled according to the Balls and Bins distribution with $m$ balls, and let $X_1,\ldots,X_n$ be $n$ i.i.d.\ random variables sampled from the distribution $\pois(\frac{m}{n})$. Let $f\colon \R^n\rightarrow \{0,1\}$ be a boolean function over the weight vector, such that the probability  $p(w_1,\ldots,w_n)=\pr[f(w_1,\ldots,w_n) = 1]$ is monotonically increasing or decreasing with the number of balls. Then $p(w_1,\ldots,w_n) \leq 2p(X_1,\ldots,X_n)$.
  \end{theorem}

The following lemma completes the proof of Lemma~\ref{lem:min_weight}, since calculation shows that for all $n \geq 1$,
\[ \frac{m}{n} \sqrt{\frac{\log (n/\log(2n))}{m/n}} = \sqrt{\frac{m\log (n/\log (2n))}{n}} \geq \sqrt{\frac{m \log n}{3n}}. \]
(In fact, the minimum of $\frac{\log (n/\log(2n))}{\log n}$ is obtained for $n = 3$, in which case it is roughly $0.47$.)

\begin{lemma}\label{lem:min_weight_eps}
Let $\lambda = \frac mn$. For any $\eps \le  \sqrt{\frac{\log\left(\frac{n}{\log(2n)}\right)}{\lambda}}$, $\pr[w_1 > \lambda(1-\eps)] \leq \frac{1}{n}$.
\end{lemma}
\begin{proof}
We define $n$ i.i.d random variables $X_1,\ldots,X_n$, sampled from the distribution $\pois(\lambda)$. We first derive a concentration bound on $\min_i X_i$, after which we will make use of Theorem~\ref{thm:poisson_approx} to obtain the desired result. By the definition of the Poisson distribution,
  \[
    \pr[\min_{i}X_i > t ] = \pr[X_1 > t]^n \leq \Pr[X_1 \neq t]^n\leq \left( 1 - e^{-\lambda} \frac{\lambda^t}{t!} \right)^n \leq \left( 1 - e^{-\lambda} \left(\frac{e\lambda}{t}\right)^t \right)^n.
  \]
  The last inequality is due to the fact that $t! \ge \left(\frac te\right)^t$, by Stirling's approximation. Setting $t = (1-\eps)\lambda$, we get
  \begin{align*}
    \label{eq:poisson_min}
   \pr[\min_i X_i > (1-\eps)\lambda] &\leq \left( 1 - e^{-\lambda} \left(\frac{e\lambda}{(1-\eps) \lambda}\right)^{(1-\eps)\lambda} \right)^n \\
   &= \left( 1 - e^{-\lambda} \left(\frac{e}{1-\eps}\right)^{(1-\eps) \lambda} \right)^n \\
   &\leq \left( 1 - e^{-\eps\lambda}e^{(1-\eps)\eps\lambda} \right)^n \\
   &= \left( 1 - e^{-\eps^2\lambda} \right)^n \leq e^{-ne^{-\eps^2\lambda}}.
  \end{align*}
  The second inequality follows from the inequality $\frac{1}{1-x} \geq e^x$, for $|x|<1$.
  The  third inequality follows from the inequality $1-x \leq e^{-x}$.

Now, for any $\eps\le \sqrt{\frac{\log\left(\frac{n}{\log(2n)}\right)}{\lambda}}$, we have
\[
e^{-ne^{-\eps^2\lambda}} \le e^{-ne^{-\log\left(\frac{n}{\log(2n)}\right)}}  = e^{-\log(2n)} = \frac{1}{2n}.
\]
A simple coupling argument shows that $\pr[\min_i w_i > (1-\eps)\lambda]$ is monotone increasing in the number of balls (here, $f(w_1,\dots,w_n)$ is 1 if and only if $\min_i w_i > (1-\eps)\lambda$). Therefore Theorem~\ref{thm:poisson_approx} holds, and we have
$$\pr[\min_i w_i > (1-\eps)\lambda] \le 2\pr[\min_i X_i > (1-\eps)\lambda] \le \frac1n,$$
which concludes the proof.
\end{proof}

\subsection{Proof of Lemma~\ref{lem:p_eps_bound}}
\pepsbound*
\begin{proof}
 The idea of the proof is to replace $w(A)$ by the weight of a random set of size $\ell-1$. A simple coupling argument shows that
\begin{eqnarray*}
\teps & \geq & \pr_{A \inR \setbinom{N}{\ell-1}}\left[w(A)+w_1 \geq q : w_1= \frac{m}{n}-\eps\sqrt{\frac{m\log n}{n}}\right] \\
			& = & \pr\left[\bin(m,\tfrac{\ell-1}{n}) \geq  \frac{(\ell-1)m}{n} + \eps\sqrt{\frac{m\log n}{n}}\right],
\end{eqnarray*}
 using the second part of Lemma~\ref{lem:prob-est}.

As in the proof of Lemma~\ref{lem:relating-ell-ell-1}, since $m \geq 4n^3$, we can use the Berry--Esseen theorem to estimate the latter expression up to an additive error of $\frac{1}{n}$:
\begin{eqnarray*}
\teps 	& \geq & \pr\left[\bin(m,\tfrac{\ell-1}{n}) \geq  \frac{(\ell-1)m}{n} + \eps\sqrt{\frac{m\log n}{n}}\right] \\
				& \geq & \pr\left[\nor(\tfrac{(\ell-1)m}{n},\tfrac{(\ell-1)m}{n}(1-\tfrac{\ell-1}{n}))  \geq \frac{(\ell-1)m}{n} + \eps\sqrt{\frac{m\log n}{n}}\right] - \frac{1}{n}.
\end{eqnarray*}
In order to estimate the latter probability, we use the bound $\pr[\nor(0,1) \geq x] \geq 1/2 - \frac{x}{\sqrt{2\pi}}$ (for $x \geq 0$), which follows from $\pr[\nor(0,1) \geq 0] = 1/2$ and the fact that the density of $\nor(0,1)$ is bounded by $1/\sqrt{2\pi}$. In our case,
\[
x = \eps\sqrt{\frac{m\log n}{n}} \Big/ \sqrt{\tfrac{(\ell-1)m}{n}(1-\tfrac{\ell-1}{n})} \leq
\eps\sqrt{\frac{m\log n}{n}} \Big/ \sqrt{\gamma^2 m} =
\eps\sqrt{\frac{\log n}{\gamma^2n}}.
\]
Therefore
\[
\teps \geq \frac{1}{2} - \frac{\eps}{2\pi\gamma} \sqrt{\frac{\log n}{n}} - \frac{1}{n}.
\]
\end{proof}

\section{Missing Proofs From Section~\ref{sec:balls_bins_exponential}}
\label{app:balls-and-bins-exponential}
\subsection{Proof of Lemma~\ref{lem:super-increasing}}
\superincreasing*
\begin{proof}
	The proof uses Bernstein's inequality with a subsequent application of the union bound.
	Consider a sequence $w_1 \leq w_2 \leq \cdots \leq w_n$. The sequence is clearly super-increasing if for every $i=2,\ldots,n$, $w_{i}/w_{i-1} \geq 2$, and $w_1 > 0$. We now lower bound the probability of this event, by upper-bounding the probability of the following bad events: $E_i$ is the event that $w_i < 2w_{i-1}$ (for $i=2,\ldots,n$), and $E_1$ is the event that $w_1 = 0$. A union bound shows that the sequence $\mathbf{w}$ is super-increasing with probability at least $1 - \sum_{i=1}^n \Pr[E_i]$.
\yf{Below you claimed that the probability $\pr[\frac{w_i}{w_{i-1}} < 2]$ is non-increasing in $i$, but this is not at all obvious!}

First note that the probability that voter $j$ votes for candidate $i$ is equal to
\[ p_i = \frac{\rho^{n-i}}{\sum_{i=1}^n \rho^{n-i}} = \frac{\rho^{n-i}(1-\rho)}{1 - \rho^n} = \Theta(\rho^{n-i}). \]
Bounding the probability of $E_1$ is easy:
\[
 \Pr[E_1] = (1-p_1)^m \leq e^{-p_1 m} = e^{-\Theta(\rho^{n-1} m)}.
\]

In order to bound the probability of $E_i$ for $i \neq 1$, consider the random variable $X = 2w_{i-1} - w_i$. This random variable is a sum of $m$ i.i.d. random variables $X^{(1)},\ldots,X^{(m)}$ corresponding to the different voters with the following distribution:
\[ X^{(j)} = \begin{cases} 2 & \text{ w.p. } p_{i-1}, \\ -1 & \text{ w.p. } p_i, \\ 0 & \text{ w.p. } 1-p_{i-1}-p_i. \end{cases} \]
Using the identity $p_{i-1} = \rho p_i$, the moments of $X$ are
\begin{align*}
\E[X] &= m\E[X^{(j)}] = (2\rho-1) p_im = \Theta((2\rho-1) \rho^{n-i}m), \\
\var[X] &= m(\E[X^{(j)2}] - \E[X^{(j)}]^2) = (4\rho+1) p_im - (2\rho-1)^2 p_i^2m = O(\rho^{n-i}m).
\end{align*}
Since $|X^{(j)} - \E[X^{(j)}]| = O(1)$, Bernstein's equality gives
\begin{align*}
\Pr[E_i] &= \Pr[X > 0] \\ &\leq \exp -\frac{\tfrac{1}{2} \E[X]^2}{\var[X] + O(\E[X])} \\ &=
\exp -\frac{\Theta((2\rho-1)^2 \rho^{2(n-i)}m^2)}{O(\rho^{n-i}m)} \\ &=
\exp -\Omega((2\rho-1)^2 \rho^{n-i}m).
\end{align*}

Summarizing,
\[ \sum_{i=1}^n \Pr[E_i] \leq e^{-\Theta(\rho^{n-1} m)} + \sum_{i=2}^n e^{-\Omega((2\rho-1)^2 \rho^{n-i}m)}. \]
When $m \geq C\rho^{-n} (2\rho-1)^{-2} \log n$ for an appropriate $C$, all the terms are $O(1/n^2)$, and so the total error probability is $O(1/n)$, proving the first part of lemma. As $m\to\infty$, all the terms tend to~$0$, and so the total error probability tends to~$0$, proving the second part of the lemma.
\end{proof}
\subsection{Proof of Lemma~\ref{lem:super-increasing2}}
\superincreasingtwo*
\begin{proof}
Suppose that $\mathbf{w}$ is super-increasing.
Lemma~\ref{lem:super-increasing-rep} implies that $I^{\mathbf{w}}(P) = (\tilde w(P), \tilde w(P^+)]$, since both $\mathbf{p}$ and $\mathbf{w}$ are super-increasing (a priori, it could be that $P^+$ would have different values when defined with respect to $\mathbf{p}$ and to $\mathbf{w}$).
The idea of the proof is to show that with high probability, $mT \in I^{\mathbf{w}}(P)$, and then the lemma follows from Proposition~\ref{pro:super-increasing-formula}. We do that by upper-bounding the probability of the following two bad events: $\tilde w(P) \geq mT$ and $\tilde w(P^+) < mT$.

The random variable $\tilde w(P)$ is a sum of $m$ i.i.d. indicator random variables which are $1$ with probability $\tilde p(P)$. Therefore $\E[\tilde w(P)] = m \tilde p(P)$. Hoeffding's inequality shows that
\[
\Pr[ \tilde w(P) \geq mT ] \leq \Pr[ \tilde w(P) \geq \E[\tilde w(P)] + m\Delta ] \leq e^{-2\Delta^2m}.
\]
Similarly $\Pr[ \tilde w(P^+) < mT ] \leq e^{-2\Delta^2m}$. When $\Delta \geq \sqrt{\log n/m}$, both error probabilities are at most $1/(nm)^2$.
\end{proof}


\section{Missing Proofs from Section~\ref{sec:super-increasing}}
\subsection{A Greedy Algorithm for Finding $\pseqall(q)$} \label{sec:si-greedy}
Given a point $q \in (0,w(N)]$ and a vector of super-increasing weights $\vec w$, it is possible to find $\pseqall(q)$ in time $O(n)$.

\begin{lemma} \label{lem:super-increasing-greedy}
Algorithm~\ref{alg:find set} calculates $\pseqall(q)$ in linear time.
\end{lemma}
\begin{proof}
As stated, the algorithm does not in fact run in linear time, but it is easy to modify it so that it does run in linear time. It remains to prove that it calculates $\pseqall(q)$ correctly.

Let $\pseqall(q) = \pseq_0,\ldots,\pseq_r$, so that $\pseqall(q)^- = \pseq_0,\ldots,\pseq_{r-1},\pseq_r+1,\ldots,n$.
Denote by $\pseqall_i$ the value of $\pseqall$ in the algorithm after $i$ iterations of the loop. We prove by induction on $i$ that $\pseqall_i = \pseqall(q) \cap \{1,\ldots,i\}$, which shows that the algorithm returns $\pseqall(q)$. The inductive claim trivially holds for $i = 0$. Assuming that $\pseqall_{i-1} = \pseqall(q) \cap \{1,\ldots,i-1\}$, we now prove that $\pseqall_i = \pseqall(q) \cap \{1,\ldots,i\}$. We consider two cases: $i \notin \pseqall(q)$ and $i \in \pseqall(q)$.
If $i \notin \pseqall(q)$ then $q \leq w(\pseqall(q)) = w(\pseqall_{i-1}) + w(\pseqall(q) \cap \{i,\ldots,n\}) \leq w(\pseqall_{i-1}) + w(\{i+1,\ldots,n\})$, and so line~5 does not get executed.
Suppose now that $i \in \pseqall(q)$. If $\pseq_r = i$ then $q > w(\pseqall(q)^-) = w(\pseqall_{i-1}) + w(\{i+1,\ldots,n\})$, and so line~5 gets executed. If $\pseq_r > i$ then $q > w(\pseqall(q)^-) \geq w(\pseqall_{i-1}) + w_i > w(\pseqall_{i-1}) + w(\{i+1,\ldots,n\})$, since $\vec w$ is super-increasing, and so line~5 gets executed in this case as well.
\end{proof}

\begin{algorithm}
\caption{An algorithm for finding $\pseqall(q)$}
\label{alg:find set}
\begin{algorithmic}[1]
\Procedure{Find-Set}{$\vec w,q$}
	\State $\pseqall \gets \emptyset$
	\For{$i \gets 1$ \textbf{to} $n$}
	  \If{$q > w(\pseqall \cup \{i+1,\ldots,n\})$}
	    \State $\pseqall \gets \pseqall \cup \{i\}$
	  \EndIf
	\EndFor
	\Return $\pseqall$
\EndProcedure
\end{algorithmic}

\end{algorithm}

\subsection{Finding $\pseqall(q)$ when Weights are $w_i = d^{n-i}$} \label{sec:si-dary}
\begin{lemma} \label{lem:super-increasing-d-ary}
Suppose $w_i = d^{n-i}$ for some integer $d \geq 2$, and let $q \in (0,w(N)]$. Write $\lceil q \rceil$ in base $d$: $\lceil q \rceil = (t_1\ldots t_n)_d$. If the base $d$ representation only consists of the digits $0$ and $1$ then $\pseqall(q) = \{ i \in N : t_i = 1 \}$. Otherwise, let $\ell$ be the minimal index such that $t_\ell > 1$, and let $k < \ell$ be the maximal index less than $\ell$ satisfying $t_k = 0$ (the proof shows that such an index exists). Then $\pseqall(q) = \{ i \in \{1,\ldots,k-1\} : t_i = 1 \} \cup \{k\}$.
\end{lemma}
\begin{proof}
 Suppose first that $t_i \in \{0,1\}$ for all $i \in N$, and let $Q(q) = \{ i \in N : t_i = 1 \}$. Since $\lceil q \rceil \geq 1$, $Q(q) \neq \emptyset$. Lemma~\ref{lem:super-increasing-rep} shows that $w(Q(q)^-) < w(Q(q))$ and so $q = w(Q(q)) \in (w(Q(q)^-), w(Q(q))]$, showing that $\pseqall(q) = Q(q)$.

 Suppose next that $\ell$ is the minimal index such that $t_\ell > 1$. If $t_k = 1$ for all $k < \ell$ then
\[ q > \lceil q \rceil - 1 \geq \sum_{j=1}^{\ell-1} w_j + 2w_\ell - 1 \geq w(N), \]
since the fact that the $w_i$ are integral and super-increasing implies that
\[ w_\ell \geq \sum_{j=\ell+1}^n w_j + 1. \]
 We conclude that the maximal index $k < \ell$ satisfying $t_k = 0$ exists. Let $Q(q) = \{ i \in \{1,\ldots,k-1\} : t_i = 1 \} \cup \{k\}$. On the one hand,
\[ q \leq \lceil q \rceil \leq \sum_{j \in Q(q) \setminus \{k\}} w_j + (d-1) \sum_{j=k+1}^n w_j < w(Q(q)). \]
 On the other hand,
\begin{align*}
 q > \lceil q \rceil - 1 &\geq \sum_{j \in Q(q) \setminus \{k\}} w_j + \sum_{j=k+1}^{\ell-1} w_j + 2w_\ell - 1 \\ &\geq \sum_{j \in Q(q) \setminus \{k\}} w_j + \sum_{j=k+1}^n w_j = w(Q(q)^{-}).
\end{align*}
 Therefore $\pseqall(q) = Q(q)$.
\end{proof}

\subsection{Proof of Theorem~\ref{thm:super-increasing-formula}}
First, we recall the statement of Theorem~\ref{thm:super-increasing-formula}.
\superincreasingformula*
\begin{proof}
Lemma~\ref{lem:super-increasing-P} shows that $\sv_i^{\vec w}(q) = \sv_i^{\vec b}(\beta(A(q)))$, where $\vec b = 2^{n-1},\ldots,1$. Therefore we can assume without loss generality that $\vec w = 2^{n-1},\ldots,1$, i.e., $w_i = 2^{n-i}$, and that $q = \sum_{j \in \pseqall(q)} w_j$.

Recall that $\sv_i(q)$ is the probability that $w(P_i(\pi)) \in [q-w_i,q)$, where $\pi$ is chosen randomly from $\perms{n}$, and $P_i(\pi)$ is the set of predecessors of~$i$ in~$\pi$.
The idea of the proof is to consider the maximal $\tau \in \{1,\ldots,r+1\}$ such that $\pseq_t \in P_i(\pi)$ for all $t < \tau$. We will show that when $i \notin \pseqall(q)$, each possible value of $\tau(\pi)$ corresponds to one summand in the expression for $\shapley_i(q)$. When $i \in \pseqall(q)$, say $i = \pseq_s$, we will show that the events that $i$ is pivotal with respect to $q$ and that $i$ is pivotal with respect to $q - w_i$ are disjoint, and their union is an event having probability $1/\pseq_s\binom{\pseq_s-1}{s}$.

 Suppose that $i$ is pivotal for $\pi$ and $\tau(\pi) = \tau$. We
 start by showing that $\tau \leq r$, ruling out the case $\tau =
 r+1$. If $\tau = r+1$ then by definition
\[ w(P_i(\pi)) \geq \sum_{j \in \pseqall(q)} w_j = q, \]
 contradicting the assumption $w(P_i(\pi)) < q$. Therefore $\tau \leq
 r$, and so $\pseq_\tau$ is well-defined. We claim that if $k \in P_i(\pi)$
 for some agent $k < \pseq_\tau$ then $k \in \pseqall(q)$. Indeed, otherwise
\[
 w(P_i(\pi)) \geq \sum_{t=0}^{\tau-1} w_{\pseq_t} + w_k \geq \sum_{t=0}^{\tau-1} w_{\pseq_t} + w_{\pseq_\tau-1} >
\sum_{t=0}^{\tau-1} w_{\pseq_t} + \sum_{j=\pseq_\tau}^n w_j \geq q,
\]
 again contradicting $w(P_i(\pi)) < q$ (the third inequality made use
 of the fact that $\vec w$ is super-increasing).

 Furthermore, we claim that $\pseq_\tau \geq i$. Otherwise,
\[
 w(P_i(\pi)) \leq \sum_{t=0}^{\tau-1} w_{\pseq_t} + \sum_{j=\pseq_\tau+1}^n w_j - w_i <
 \sum_{t=0}^{\tau} w_{\pseq_t} - w_i \leq q - w_i,
\]
 contradicting the assumption $w(P_i(\pi)) \geq q - w_i$.

 Summarizing, we have shown that $\tau \leq r$, $\pseq_\tau \geq i$ and
\begin{equation} \label{eq:shapley-condition-1}
 P_i(\pi) \cap \{1,\ldots,\pseq_\tau\} = \{\pseq_0,\ldots,\pseq_{\tau-1}\}.
\end{equation}
 Denote this event $E_\tau$, and call a $\tau \leq r$ satisfying $\pseq_\tau \geq i$ \emph{legal}.

 Suppose first that $i \notin \pseqall(q)$. We have shown above that
 if $i$ is pivotal then $E_\tau$ happens for some legal $\tau$. We
 claim that the converse is also true. Indeed, given $E_\tau$ defined
 with respect to a permutation $\pi$, and for some legal $\tau$, the
 weight of $P_i(\pi)$ can be bounded as follows.
\[
 \sum_{t=0}^{\tau-1} w_{\pseq_t} \leq w(P_i(\pi)) \leq \sum_{t=0}^{\tau-1} w_{\pseq_t} + \sum_{j=\pseq_{\tau}+1}^n w_j < \sum_{t=0}^{\tau} w_{\pseq_t}.
\]
The second inequality follows from the definition of $\tau$, whereas
the third inequality follows as before from the definition of a super-increasing sequence.
 The upper bound is clearly at most $q$, and the lower bound satisfies
\[
\sum_{t=0}^{\tau-1} w_{\pseq_t} \geq q - \sum_{j=\pseq_\tau}^n w_j > q - w_{\pseq_\tau-1} \geq q - w_i,
\]
 since $i < \pseq_\tau$.

 It remains to calculate $\Pr[E_\tau]$. The event $E_\tau$ states that the restriction of $\pi$ to $\{1,\ldots,\pseq_\tau\}$ consists of the elements $\{\pseq_0,\ldots,\pseq_{\tau-1}\}$ in some order, followed by $i$ (recall that $i \leq \pseq_\tau$). For each of the $\tau!$ possible orders, the probability of this is $1/\pseq_\tau\cdots(\pseq_\tau-\tau) = (\pseq_\tau-\tau-1)!/\pseq_\tau!$, and so
\begin{equation} \label{eq:shapley-formula}
 \Pr[E_\tau] = \frac{\tau!(\pseq_\tau-\tau-1)!}{\pseq_\tau!} = \frac{1}{\pseq_\tau\binom{\pseq_\tau-1}{\tau}}.
\end{equation}
 Summing over all legal $\tau$, we obtain the formula in the statement of the theorem. This completes the proof in the case $i \notin \pseqall(q)$.

 Suppose next that $i \in \pseqall(q)$, say $i = \pseq_s$. Since $\pseq_\tau \geq \pseq_s = i$ while $i \notin P_i(\pi)$, we deduce that $\tau = s$. Therefore the event $E_s$ happens.
 Conversely, when $E_s$ happens,
\[ w(P_i(\pi)) \leq \sum_{t=0}^{s-1} w_{\pseq_t} + \sum_{j=\pseq_s+1}^n w_j < \sum_{t=0}^s w_{\pseq_t} \leq q. \]
 Therefore $i$ is pivotal (with respect to $q$) if and only if $E_s$ happens and $w(P_i(\pi)) \geq q - w_i$.

 It is easy to check that $\pseqall(q - w_i) = \pseqall(q) \setminus \{i\} = \pseq_0,\ldots,\pseq_{s-1},\pseq_{s+1},\ldots,\pseq_r$. The argument above shows that if $i$ is pivotal with respect to $q - w_i$ then for some $\tau' \geq s+1$,
\[ P_i(\pi) \cap \{1,\ldots,\pseq_{\tau'}\} = \{\pseq_0,\ldots,\pseq_{s-1},\pseq_{s+1},\ldots,\pseq_{\tau'-1}\}. \]
 In particular, the event $E_s$ happens.
 Conversely, when $E_s$ happens,
 \[ w(P_i(\pi)) \geq \sum_{t=0}^{s-1} w_{\pseq_t} \geq q - w_{\pseq_s} - \sum_{j=\pseq_s+1}^n w_j > (q - w_{\pseq_s}) - w_{\pseq_s}. \]
 Therefore $i$ is pivotal with respect to $q - w_i$ if and only if $E_s$ happens and $w(P_i(\pi)) < q - w_i$. We conclude that
\[
 \Pr[w_i \text{ is pivotal with respect to } q] = \Pr[E_s] - \Pr[w_i \text{ is pivotal with respect to } q - w_i].
\]
 Above we have calculated $\Pr[E_s] = 1/\pseq_s\binom{\pseq_s-1}{s}$, and we obtain the formula in the statement of the theorem.
\end{proof}

\subsection{Proof of Lemma~\ref{lem:super-increasing-adjacent}}
To prove Lemma~\ref{lem:super-increasing-adjacent}, we will need some combinatorial identities.

\begin{lemma} \label{lem:super-increasing-identity}
 Let $p,t$ be integers satisfying $p > t \geq 1$. Then
\[  \frac{1}{p\binom{p-1}{t}} + \frac{1}{p\binom{p-1}{t-1}} = \frac{1}{(p-1)\binom{p-2}{t-1}}. \]
\end{lemma}
\begin{proof}
The proof is a simple calculation:
\begin{align*}
 \frac{1}{p\binom{p-1}{t}} + \frac{1}{p\binom{p-1}{t-1}} &=
 \frac{t!(p-t-1)! + (t-1)!(p-t)!}{p!} \\ &= \frac{(t-1)!(p-t-1)![t+(p-t)]}{p!} = \frac{(t-1)!(p-t-1)!}{(p-1)!} = \frac{1}{(p-1)\binom{p-2}{t-1}}.
\end{align*}
\end{proof}

\begin{lemma} \label{lem:super-increasing-sum-identity}
 Let $p,t,k$ be integers satisfying $p > t \geq 0$ and $k \geq 0$. Then
\[ \frac{1}{p\binom{p-1}{t}} - \sum_{\ell=1}^k \frac{1}{(p+\ell)\binom{p+\ell-1}{t+\ell-1}} = \frac{1}{(p+k)\binom{p+k-1}{t+k}}. \]
 In particular,
\[ \frac{1}{p\binom{p-1}{t}} = \sum_{\ell=1}^\infty \frac{1}{(p+\ell)\binom{p+\ell-1}{t+\ell-1}}. \]
\end{lemma}
\begin{proof}
 The proof is by induction on $k$. If $k = 0$ then there is nothing to prove. For $k > 0$ we have
\[
 \frac{1}{p\binom{p-1}{t}} - \sum_{\ell=1}^k \frac{1}{(p+\ell)\binom{p+\ell-1}{t+\ell-1}} =
 \frac{1}{(p+k-1)\binom{p+k-2}{t+k-1}} - \frac{1}{(p+k)\binom{p+k-1}{t+k-1}} = \frac{1}{(p+k)\binom{p+k-1}{t+k}},
\]
 using Lemma~\ref{lem:super-increasing-identity}. The second expression of
 the lemma follows from rearranging the first formula and taking the
 limit $k \rightarrow \infty$.
\end{proof}
We are now ready to prove Lemma~\ref{lem:super-increasing-adjacent}. First, recall the statement of the lemma.
\superincreasingadjacent*
\begin{proof}
 For the first item, since $i+1 \notin \pseqall(q)$ then $\pseq_t > i$ iff $\pseq_t > i+1$, and so
\[
 \shapley_i(q) = \sum_{\substack{t \in \{0,\ldots,r\}\colon\\ \pseq_t > i}} \frac{1}{\pseq_t\binom{\pseq_t-1}{t}} =
 \sum_{\substack{t \in \{0,\ldots,r\}\colon\\ \pseq_t > i+1}} \frac{1}{\pseq_t\binom{\pseq_t-1}{t}} = \shapley_{i+1}(q).
\]

 For the second item, suppose that $i+1 = \pseq_s$. We have
\begin{align*}
 \shapley_i(q) - \shapley_{i+1}(q) &= \sum_{t=s}^r \frac{1}{\pseq_t\binom{\pseq_t-1}{t}} -
 \left[ \frac{1}{\pseq_s\binom{\pseq_s-1}{s}} - \sum_{t=s+1}^r \frac{1}{\pseq_t\binom{\pseq_t-1}{t-1}} \right] \\ &=
 \sum_{t=s+1}^r \left[ \frac{1}{\pseq_t\binom{\pseq_t-1}{t}} + \frac{1}{\pseq_t\binom{\pseq_t-1}{t-1}} \right] =
 \sum_{t=s+1}^r \frac{1}{(\pseq_t-1) \binom{\pseq_t-2}{t-1}},
\end{align*}
 using Lemma~\ref{lem:super-increasing-identity}.
 Therefore $\shapley_i(q) \geq \shapley_{i+1}(q)$, with equality if and only if $s = r$.

 For the third item, suppose that $i = \pseq_s$. We have
\begin{align*}
 \shapley_i(q) - \shapley_{i+1}(q) &= \frac{1}{\pseq_s\binom{\pseq_s-1}{s}} - \sum_{t=s+1}^r \frac{1}{\pseq_t\binom{\pseq_t-1}{t-1}} - \sum_{t=s+1}^r \frac{1}{\pseq_t\binom{\pseq_t-1}{t}} \\ &=
 \frac{1}{\pseq_s\binom{\pseq_s-1}{s}} - \sum_{t=s+1}^r \frac{1}{(\pseq_t-1)\binom{\pseq_t-2}{t-1}},
\end{align*}
 using Lemma~\ref{lem:super-increasing-identity}. The same lemma also implies that the expression $1/p\binom{p}{t-1}$ is decreasing in $p$. Since $i+1 \notin \pseqall(q)$, if $\pseq_{s+1}$ exists then $\pseq_{s+1} \geq \pseq_s + 2$, and in general $\pseq_{s+\ell} \geq \pseq_s + \ell + 1$. Therefore
\[
 \shapley_i(q) - \shapley_{i+1}(q) \geq
 \frac{1}{\pseq_s\binom{\pseq_s-1}{s}} - \sum_{\ell=1}^{r-s} \frac{1}{(\pseq_s+\ell)\binom{\pseq_s+\ell-1}{s+\ell-1}} =
 \frac{1}{(\pseq_s+r-s)\binom{\pseq_s+r-s-1}{r}} > 0,
\]
 using Lemma~\ref{lem:super-increasing-sum-identity}.

 For the fourth item, suppose that $i = \pseq_s$. We have
\begin{align*}
 \shapley_i(q) - \shapley_{i+1}(q) &=
 \left[\frac{1}{\pseq_s\binom{\pseq_s-1}{s}} - \sum_{t=s+1}^r \frac{1}{\pseq_t\binom{\pseq_t-1}{t-1}}\right] -
 \left[\frac{1}{\pseq_{s+1}\binom{\pseq_{s+1}-1}{s+1}} - \sum_{t=s+2}^r \frac{1}{\pseq_t\binom{\pseq_t-1}{t-1}}\right] \\ &=
 \frac{1}{\pseq_s\binom{\pseq_s-1}{s}} - \frac{1}{\pseq_{s+1}\binom{\pseq_{s+1}-1}{s+1}} - \frac{1}{\pseq_{s+1}\binom{\pseq_{s+1}-1}{s}} = 0,
\end{align*}
 using Lemma~\ref{lem:super-increasing-identity} together with $\pseq_{s+1} = \pseq_s+1$.
\end{proof}

\subsection{Proof of Lemma~\ref{lem:super-increasing-jumps}}
First, let us recall the statement of Lemma~\ref{lem:super-increasing-jumps}.
\superincreasingjumps*

\begin{proof}
Define $\shapley_+ = \shapley_i(w(P))$ and $\shapley_- = \shapley_i(w(P^-))$.
Let $P = \pseq_0,\ldots,\pseq_r$. We have $P^- = \pseq_0,\ldots,\pseq_{r-1},\pseq_r+1,\ldots,n$.

Suppose first that $i > \pseq_r$, and let $s$ be the index of $i$ in the sequence $P^-$. According to Theorem~\ref{thm:super-increasing-formula}, $\shapley_+ = 0$ and
\[
 \shapley_- = \frac{1}{i\binom{i-1}{s}} - \sum_{\ell=1}^{n-i} \frac{1}{(i+\ell)\binom{i+\ell-1}{s+\ell-1}} = \frac{1}{n\binom{n-1}{s+n-i}}.
\]
We see that $i \in P^-$ and $\shapley_- > \shapley_+$.
Furthermore, $|\shapley_+ - \shapley_-| \leq \frac{1}{n(n-1)}$ unless $s+n-i \in \{0,n-1\}$. If $s+n-i = 0$ then $s = 0$ and $i = n$, implying $P^- = \{n\}$ and so $P = \{n-1\}$. If $s+n-i = n-1$ then $s = i-1$ and so $P^- = 1,\ldots,n$, which is impossible.

Suppose next that $i = \pseq_r$. According to the theorem,
\[
 \shapley_+ - \shapley_- = \frac{1}{i\binom{i-1}{r}} - \sum_{\ell=1}^{n-i} \frac{1}{(i+\ell)\binom{i+\ell-1}{r+\ell-1}} = \frac{1}{n\binom{n-1}{r+n-i}}.
\]
We see that $i \notin P^-$ and $\shapley_+ > \shapley_-$.
Furthermore, $|\shapley_+ - \shapley_-| \leq \frac{1}{n(n-1)}$ unless $r+n-i \in \{0,n-1\}$. If $r+n-i = 0$ then $r = 0$ and $i = n$, and so $P = \{n\}$. If $r + n-i = n$ then $r = i-1$ and so $P = 1,\ldots,i$.

Finally, suppose that $i < \pseq_r$. If $i \notin P$ then
\[
 \shapley_+ - \shapley_- = \frac{1}{\pseq_r\binom{\pseq_r-1}{r}} - \sum_{\ell=1}^{n-\pseq_r} \frac{1}{(\pseq_r+\ell)\binom{\pseq_r+\ell-1}{r+\ell-1}} = \frac{1}{n\binom{n-1}{r+n-\pseq_r}}.
\]
We see that $i \notin P^-$ and $\shapley_+ > \shapley_-$.
Furthermore, $|\shapley_+ - \shapley_-| \leq \frac{1}{n(n-1)}$ unless $r+n-\pseq_r \in \{0,n-1\}$. If $r+n-\pseq_r = 0$ then $r = 0$ and $\pseq_r = n$, and so $P = \{n\}$. If $r+n-\pseq_r = n-1$ then $\pseq_r = r+1$, which implies $P = \{1,\ldots,r+1\}$. However, this contradicts the assumption $i \notin P$.

If $i < \pseq_r$ and $i \in P$ then
\[
 \shapley_- - \shapley_+ = \frac{1}{\pseq_r\binom{\pseq_r-1}{r-1}} - \sum_{\ell=1}^{n-\pseq_r} \frac{1}{(\pseq_r+\ell)\binom{\pseq_r+\ell-1}{r+\ell-2}} = \frac{1}{n\binom{n-1}{r+n-\pseq_r-1}}.
\]
We see that $i \in P^-$ and $\shapley_- > \shapley_+$.
Furthermore, $|\shapley_+ - \shapley_-| \leq \frac{1}{n(n-1)}$ unless $r+n-\pseq_r-1 \in \{0,n-1\}$. If $r+n-\pseq_r-1 = 0$ then $r = 1$ and $\pseq_r = n$, and so $P = \{i,n\}$. If $r+n-\pseq_r-1 = n-1$ then $\pseq_r = r$, which is impossible.
\end{proof}
\subsection{A Note on the Limiting Behavior of the Shapley Value under Super-Increasing Weights}\label{app:super-increasing-limits}
Given a super-increasing sequence $w_1,\ldots,w_n$ (where again, $w_1 > w_2 > \dots> w_n$) and some $m \in N$, let us write $\vec w|_m$ for $(w_1,\dots,w_m)$ and $[m]$ for $\{1,\dots,m\}$. We write $\sv_i^{\vec w|_m}(q)$ for the Shapley value of agent $i \in [m]$ in the weighted voting game in which the set of agents is $[m]$, the weights are $\vec w|_m$, and the quota is $q$. We also write $\pseqall|_m(q)$ for the set $P \subseteq [m]$ such that $q \in (w|_m(P^-),w|_m(P)]$.

The following lemma relates $\shapley^{\vec w}_i(q)$ and $\shapley^{\vec w|_m}_i(q)$.

\begin{restatable}{lem}{superincreasingextend} \label{lem:super-increasing-extend}
 Let $m \in N$ and $i \in [m]$, and let $q \in (0,w([m])]$. Then
\jo{the following line would not compile in its original form. I
  copied it and fixed the bug, to the best of my understanding. The
  original equation is commented out in the .tex file. Please check.}
\[\shapley^{\vec w|_m}_i(q) = \shapley_i^{\vec w}(w(\pseqall|_m(q))). \]
\end{restatable}
\begin{proof}
Theorem~\ref{thm:super-increasing-formula} provides a function $\Phi$ such that $\shapley_i^{\vec w|_m}(q) = \Phi(\pseqall|_m(q))$ and $\shapley_i^{\vec w}(w(\pseqall|_m(q))) = \Phi(\pseqall(w(\pseqall|_m(q)))) = \Phi(\pseqall|_m(q))$. We conclude that the Shapley values coincide.
\end{proof}

Therefore the plot of $\shapley^{\vec w|_m}_i$ can be readily obtained from that of $\shapley^{\vec w}_i$. This suggests looking at the limiting case of an \emph{infinite} super-increasing sequence $(w_i)_{i=1}^\infty$, which is a sequence satisfying $w_i > 0$ and $w_i \geq \sum_{j=i+1}^\infty w_j$
for all $i \geq 1$. The super-increasing condition implies that the sequence sums to some value $w(\infty) \leq 2w_1$. Lemma~\ref{lem:super-increasing-extend} suggests how to define $\shapley_i$ in this case: for $q \in (0,w(\infty))$ and $i \geq 1$, let
\[ \shapley^{(\infty)}_i(q) = \lim_{n\to\infty} \shapley^{\vec w|_n}_i(q). \]
We show that the limit exists by providing an explicit formula for it, as given in Theorem~\ref{thm:super-increasing-limit}, which is proved in the appendix. In the theorem, we consider possibly infinite subsets $P = \{a_0,\ldots,a_r\}$ of the positive integers, ordered in increasing order; when $r = \infty$, the subset is infinite. Also, the notation $\{a,\ldots,\infty\}$ (or $\{a,\ldots,r\}$ when $r=\infty$) means all integers larger than or equal to $a$.

 \begin{restatable}{thm}{superincreasinglimit} \label{thm:super-increasing-limit}
 Let $q \in (0,w(\infty))$ and let $i$ be a positive integer.
\begin{enumerate}[(a)]
 \item There exists a non-empty subset of the positive integers $P = \{a_0,\ldots,a_r\}$ such that either $q = w(P)$ or $P$ is finite and $q \in (w(P^-),w(P)]$, where $P^- = \{a_0,\ldots,a_{r-1}\} \cup \{a_r+1,\ldots,\infty\}$.
 \item The limit $\shapley^{(\infty)}_i(q) = \lim_{n\to\infty} \shapley^{\vec w|_n}_i(q)$ exists. When $i \notin P$,
\[ \shapley^{(\infty)}_i(q) = \sum_{\substack{t \in \{0,\ldots,r\}\colon\\ \pseq_t > i}} \frac{1}{\pseq_t\binom{\pseq_t-1}{t}}, \]
 and when $i \in P$, say $i = \pseq_s$, then
\[ \shapley^{(\infty)}_i(q) = \frac{1}{\pseq_s\binom{\pseq_s-1}{s}} - \sum_{\substack{t \in \{0,\ldots,r\}\colon\\ \pseq_t > i}} \frac{1}{\pseq_t\binom{\pseq_t-1}{t-1}}. \]
\end{enumerate}
\end{restatable}

Lemma~\ref{lem:super-increasing-extend} easily extends to the case $n = \infty$.

\begin{lemma} \label{lem:super-increasing-extend-inf}
 Let $m \geq 1$ be an integer, let $i \in [m]$, and let $q \in (0,w([m])]$. Then $\shapley^{\vec w|_m}_i(q) = \shapley^{(\infty)}_i(w(\pseqall|_m(q)))$.
\end{lemma}
\begin{proof}
Lemma~\ref{lem:super-increasing-extend} shows that for $n \geq m$, $\shapley^{\vec w|_m}_i(q) = \shapley^{\vec w|_n}_i(w(\pseqall|_m(q)))$, and therefore $\shapley^{\vec w|_m}_i(q) = \lim_{n\to\infty} \shapley^{\vec w|_n}_i(w(\pseqall|_m(q))) = \shapley^{(\infty)}_i(w(\pseqall|_m(q)))$.
\end{proof}

We conclude by showing that the limiting functions $\sv_i^{(\infty)}$ are continuous (see appendix for the proof).
\begin{restatable}{thm}{superincreasingcontinuous} \label{thm:super-increasing-continuous}
Let $i$ be a positive integer. The function $\shapley^{(\infty)}_i$ is continuous on $(0,w(\infty))$, and
$\lim_{q \to 0} \shapley^{(\infty)}_i(q) = \lim_{q \to w(\infty)} \shapley^{(\infty)}_i(q) = 0.$
\end{restatable}
Summarizing, we can extend the functions $\shapley^{\vec w|_n}_i$ to a continuous function $\shapley^{(\infty)}_i$ which agrees with $\shapley^{\vec w|_n}_i$ on the points $w(P)$ for $P \subseteq \{1,\ldots,n\}$. When $w_i = 2^{-i}$ then the plot of $\shapley^{(\infty)}$ has no flat areas, but when $w_i = d^{-i}$ for $d > 2$, the limiting function is constant on intervals $(w(P^-),w(P)]$. This is reflected in Figure~\ref{fig:super-increasing}.

\subsection{Proof of Theorem~\ref{thm:super-increasing-limit}}
We start with some preliminary lemmas. For a (possibly infinite) subset $P$ of the positive integers, define
\[ \beta_\infty(P) = \sum_{i \in P} 2^{-i}. \]
We have the following analog of Lemma~\ref{lem:super-increasing-rep}.

\begin{lemma} \label{lem:super-increasing-rep-inf}
 Suppose $P_1,P_2$ are two subsets of the positive integers. Then $\beta_\infty(P_1) \leq \beta_\infty(P_2)$ if and only if $w(P_1) \leq w(P_2)$. Furthermore, if $\beta_\infty(P_1) < \beta_\infty(P_2)$ then $w(P_1) < w(P_2)$.
\end{lemma}
\begin{proof}
 Suppose that $\beta_\infty(P_1) \leq \beta_\infty(P_2)$ and $P_1 \neq P_2$. Let $i = \min (P_2 \setminus P_1)$. Then
\[
w(P_2) - w(P_1) \geq w_i - \sum_{j=i+1}^\infty w_j \geq 0.
\]
Equality is only possible if $\max P_2 = i$ and $P_1 = P_2 \setminus \{i\} \cup \{i+1,\ldots,\infty\}$. However, in that case $\beta_\infty(P_1) = \beta_\infty(P_2)$.
\end{proof}

There is a subtlety involved here: we can have $\beta_\infty(P_1) = \beta_\infty(P_2)$ for $P_1 \neq P_2$. This is because dyadic rationals (numbers of the form $\frac{A}{2^B}$) have two different binary expansions. For example, $\frac12 = (0.1000\ldots)_2 = (0.0111\ldots)_2$. The lemma states (in this case) that $w(\{1\}) \geq w(\{2,3,4,\ldots\})$, but there need not be equality.

In the sequel, we will use the fact that each real $r \in (0,1)$ has a binary expansion with infinitely many $0$s (alternatively, a set $P$ such that $\beta_\infty(P) = r$ and there are infinitely many $n \notin P$), and a binary expansion with infinitely many $1$s (alternatively, a set $P$ such that $\beta_\infty(P) = r$ and there are infinitely many $n \in P$). If $r$ is not dyadic, then it has a unique binary expansion which has infinitely many $0$s and $1$s. If $r$ is dyadic, say $r=\frac12$, then it has one expansion $(0.1000\ldots)_2$ with infinitely many $0$s and another expansion $(0.0111\ldots)_2$ with infinitely many $1$s.

The following lemma, which forms the first part of Theorem~\ref{thm:super-increasing-limit}, describes the analog of the intervals $(w(P^-),w(P)]$ in the infinite case.

\begin{lemma} \label{lem:super-increasing-partition-inf}
 Let $q \in (0,w(\infty))$. There exists a non-empty subset $P$ of the positive integers such that either $q = w(P)$ or $P=\{\pseq_0,\ldots,\pseq_r\}$ is finite and $q \in (w(P^-),w(P)]$, where $P^- = \{\pseq_0,\ldots,\pseq_{r-1} \} \cup \{ \pseq_r+1,\ldots,\infty \}$.
\end{lemma}
\begin{proof}
 Since $q < w(\infty)$, for some $m$ we have $q \leq w([m])$. For $n \geq m$, let $\pseqall|_n = \pseqall|_n(q)$. Let $Q|_n$ be the subset of $[n]$ preceding $\pseqall|_n$, and let $R|_n$ be the subset of $[n+1]$ preceding $\pseqall|_n$; here ``preceding'' is in the sense of $X \mapsto X^-$. The interval $(w(Q|_n),w(\pseqall|_n)]$ splits into $(w(Q|_n),w(R|_n)] \cup (w(R|_n),w(\pseqall|_n)]$, and so $\pseqall|_{n+1} \in \{R|_n,\pseqall|_n\}$. Also $\beta_\infty(\pseqall|_{n+1}) \leq \beta_\infty(\pseqall|_n)$, with equality only if $\pseqall|_{n+1} = \pseqall|_n$.

We consider two cases. The first case is when for some integer $M$, for all $n \geq M$ we have $\pseqall|_n = \pseqall = \{\pseq_0,\ldots,\pseq_r\}$. In that case for all $n \geq M$,
\[ \sum_{t=0}^{r-1} w_{\pseq_t} + \sum_{t=\pseq_r+1}^n w_t < q \leq \sum_{t=0}^r w_{\pseq_t}, \]
and taking the limit $n \to \infty$ we obtain $q \in (w(\pseqall^-),w(\pseqall)]$.

 The other case is when $\pseqall|_n$ never stabilizes. The sequence $\beta_\infty(\pseqall|_n)$ is monotonically decreasing, and reaches a limit $b$ satisfying $b < \beta_\infty(\pseqall|_n)$ for all $n$. Since $w(\pseqall|_m) \in (w(Q|_n),w(\pseqall|_n)]$ for all integers $m \geq n \geq 1$, Lemma~\ref{lem:super-increasing-rep-inf} implies that $b \in [\beta_\infty(Q|_n),\beta_\infty(A|_n))$.

Let $L$ be a subset such that $b = \beta_\infty(L)$ and there are infinitely many $i \notin L$, and define $L|_n = L \cap [n]$. We have $b \in [\beta_\infty(L|_n),\beta_\infty(L|_n) + 2^{-n})$. Therefore $Q|_n = L|_n$, and so $q > w(Q|_n) = w(L|_n)$. Taking the limit $n \to \infty$, we deduce that $q \geq w(L)$.

If $n \notin L$ then $\pseqall|_n = Q|_n \cup \{n\}$, and so $q \leq w(\pseqall|_n) = w(L|_n) + w_n$. Since there are infinitely many such $n$, taking the limit $n \to \infty$ we conclude that $q \leq w(L)$ and so $q = w(L)$.
\end{proof}

We can now give an explicit formula for $\shapley^{(\infty)}_i$.
\superincreasinglimit*

We comment that the convergence of the sums in the theorem is guaranteed by Lemma~\ref{lem:super-increasing-sum-identity}.

\begin{proof}
The first part has been proved as Lemma~\ref{lem:super-increasing-partition-inf}, and it remains to prove the second part.

Suppose first that $P$ is finite $P$ and either $q = w(P)$ or $q \in (w(P^-),w(P)]$. For all $n \geq \max P$, $P|_n(q) = P$, and so Lemma~\ref{lem:super-increasing-extend} shows that $\shapley^{\vec w|_n}_i(q) = \shapley^{\vec w|_{\max P}}_i(q)$. Therefore the limit exists and equals the stated formula, which is the same as the one given by Theorem~\ref{thm:super-increasing-formula}.

Suppose next that $P$ is infinite and $q = w(P)$. Consider first the case in which we can also write $q = w(Q)$ for some finite $Q$, say $Q = \{q_0,\ldots,q_u\}$. Then $P = \{q_0,\ldots,q_{u-1}\} \cup \{q_u+1,q_u+2,\ldots,\infty\}$. We now consider several cases.

If $i < q_u$ and $i \notin P$ then $i \notin Q$ and
\[ \shapley^{(\infty)}_i(q) = \sum_{\substack{t \in \{0,\ldots,u\}\colon\\ q_t > i}} \frac{1}{q_t\binom{q_t-1}{t}} = \sum_{\substack{t \in \{0,\ldots,u-1\}\colon\\ q_t > i}} \frac{1}{q_t\binom{q_t-1}{t}} + \sum_{\ell = 1}^\infty \frac{1}{(q_u+\ell)\binom{q_u+\ell-1}{t+\ell-1}}, \]
using Lemma~\ref{lem:super-increasing-sum-identity}. The right-hand side is the expression we gave for $\shapley^{(\infty)}_i(w(P))$.

If $i < q_u$ and $i \in P$, say $i = q_s$, then $i \in Q$ and
\[ \shapley^{(\infty)}_i(q) = \frac{1}{i\binom{i-1}{s}} - \sum_{\substack{t \in \{0,\ldots,u\}\colon\\ q_t > i}} \frac{1}{q_t\binom{q_t-1}{t}} = \frac{1}{i\binom{i-1}{s}} - \sum_{\substack{t \in \{0,\ldots,u-1\}\colon\\ q_t > i}} \frac{1}{q_t\binom{q_t-1}{t}} - \sum_{\ell = 1}^\infty \frac{1}{(q_u+\ell)\binom{q_u+\ell-1}{t+\ell-1}}, \]
using Lemma~\ref{lem:super-increasing-sum-identity}. The right-hand side is the expression we gave for $\shapley^{(\infty)}_i(w(P))$.

If $i = q_u$ then $i \in Q$ and $i \notin P$. In that case
\[ \shapley^{(\infty)}_i(q) = \frac{1}{i\binom{i-1}{u}} = \sum_{\ell = 1}^\infty \frac{1}{(i+\ell)\binom{i+\ell-1}{u+\ell-1}}, \]
using Lemma~\ref{lem:super-increasing-sum-identity}. The right-hand side is the expression we gave for $\shapley^{(\infty)}_i(w(P))$.

Finally, if $i > q_u$ then $i \notin Q$ and $i \in P$. Suppose that $i$ is the $v$th member in $P$. In that case
\[ \shapley^{(\infty)}_i(q) = 0 = \frac{1}{i\binom{i-1}{v}} - \sum_{\ell=1}^\infty \frac{1}{(i+\ell)\binom{i+\ell-1}{v+\ell-1}}, \]
using Lemma~\ref{lem:super-increasing-sum-identity}. The right-hand side is the expression we gave for $\shapley^{(\infty)}_i(w(P))$.

It remains to consider the case in which $q$ cannot be written as $q = w(Q)$ for finite $Q$. In that case, there are infinitely many positive integers $n$ such that $n \in P$ and infinitely many such that $n \notin P$. This implies that for every positive integer $n$, $q \in (w(P \cap [n]),w(P \cap [n]) + w_n)$, and so $P|_n^{-}(q) = P \cap [n]$. Lemma~\ref{lem:super-increasing-jumps} shows that $|\shapley_n(q) - \shapley_n(w(P \cap [n]))| \leq \frac1n$. On the other hand, Theorem~\ref{thm:super-increasing-formula} readily implies that $\shapley_n(w(P \cap [n]))$ tends to the expression we gave for $\shapley^{(\infty)}_i(w(P))$. We conclude that $\shapley_n(q)$ tends to the same expression.
\end{proof}

\subsection{Proof of Theorem~\ref{thm:super-increasing-continuous}}
First, we recall the statement of Theorem~\ref{thm:super-increasing-continuous}.
\superincreasingcontinuous*
\begin{proof}
 Let $q \in (0,w(\infty))$. We start by showing that $\shapley^{(\infty)}_i$ is continuous from the right at $q$. Lemma~\ref{lem:super-increasing-partition-inf} shows that we can find a subset $P$ such that either $q = w(P)$ or $q \in (w(P^-),w(P)]$. If $q < w(P)$ then since $\shapley^{(\infty)}_i$ is constant on $(w(P^-),w(P)]$ according to Theorem~\ref{thm:super-increasing-limit}, clearly $\shapley^{(\infty)}_i$ is continuous from the right at $q$. Therefore we can assume that $q = w(P)$. Since $q < w(\infty)$, we can further assume that there are infinitely many $n \notin P$.

 Suppose that we have a sequence $q_j$ tending to $q$ strictly from the right. For each $j$ we can find a subset $P_j$ such that either $q_j = w(P_j)$ or $q_j \in (w(P_j^-),w(P_j)]$. We can assume that the second case doesn't happen by replacing $q_j$ with $w(P_j^-)$; the new sequence still tends to $q$ strictly from the right. So we can assume that $q_j = w(P_j) > w(P)$. Let $k(j) = \min (P_j \setminus P)$, and let $l(j) > k(j)$ be the smallest index larger than $k(j)$ such that $l(j) \notin P$. Then
\[ q_j - q = w(P_j) - w(P) \geq w_{k(j)} - \left(\sum_{t=k(j)+1}^\infty w_t - w_{l(j)}\right) \geq w_{l(j)}. \]
 As $j \to \infty$, $l(j) \to \infty$ and so $k(j) \to \infty$. Therefore we can assume without loss of generality that $k(j) > i$ for all $j$. Theorem~\ref{thm:super-increasing-limit} then implies that
\[ |\shapley^{(\infty)}_i(q_j) - \shapley^{(\infty)}_i(q)| \leq \sum_{s=0}^\infty \frac{1}{(k(j)+s)\binom{k(j)+s-1}{s}} = \frac{1}{k(j)-1}, \]
 using Lemma~\ref{lem:super-increasing-sum-identity}. Since $k(j) \to \infty$, $\shapley^{(\infty)}_i(q_j) \to \shapley^{(\infty)}_i(q)$.

 We proceed to show that $\shapley^{(\infty)}_i$ is continuous from the left at $q$. Lemma~\ref{lem:super-increasing-partition-inf} shows that we can find a subset $P$ such that either $q = w(P)$ or $q \in (w(P^-),w(P)]$. In the second case, since $\shapley^{(\infty)}_i$ is constant on $(w(P^-),w(P)]$ according to Theorem~\ref{thm:super-increasing-limit}, clearly $\shapley^{(\infty)}_i$ is continuous from the left at $q$. Therefore we can assume that $q = w(P)$. Since $q > 0$, we can further assume that there are infinitely many $n \in P$.

 Suppose that we have a sequence $q_j$ tending to $q$ strictly from the left. For each $j$ we can find a subset $P_j$ such that either $q_j = w(P_j)$ or $q_j \in (w(P_j^-),w(P_j)]$, and in both cases $q_j \leq w(P_j) < w(P)$. Let $k(j) = \min (P \setminus P_j)$, and let $l(j) > k(j)$ be the smallest index larger than $k(j)$ such that $l(j) \in P$. Then
\[ q - q_j \geq w(P) - w(P_j) \geq w_{k(j)} + w_{l(j)} - \sum_{t=k(j)+1}^\infty w_t \geq w_{l(j)}. \]
 At this point we can prove that $\shapley^{(\infty)}_i(q_j) \to \shapley^{(\infty)}_i(q)$ as in the preceding case.

 It remains to show that $\lim_{q \to 0} \shapley^{(\infty)}_i(q) = \lim_{q \to w(\infty)} \shapley^{(\infty)}_i(q) = 0$. We start by showing that $\lim_{q \to 0} \shapley^{(\infty)}_i(q) = 0$. Let $q_j$ be a sequence tending to $0$ strictly from the right. As before, we can assume that $q_j = w(P_j)$ for each $j$. Let $k(j) = \min P_j$. Since $q_j \geq w_{k(j)}$, $k(j) \to \infty$. Therefore we can assume without loss of generality that $k(j) > i$ for all $j$. Theorem~\ref{thm:super-increasing-limit} then implies that
\[
 \shapley^{(\infty)}_i(q_j) \leq \sum_{s=0}^\infty \frac{1}{(k(j)+s)\binom{k(j)+s-1}{s}} = \frac{1}{k(j)-1},
\]
 using Lemma~\ref{lem:super-increasing-sum-identity}. Since $k(j) \to \infty$, $\shapley^{(\infty)}_i(q_j) \to 0$.

 We finish the proof by showing that $\lim_{q \to w(\infty)} \shapley^{(\infty)}_i(q) = 0$. Let $q_j$ be a sequence tending to $M$ strictly from the left. As before, we can find subsets $P_j$ such that $q_j \leq w(P_j)$ and $\shapley^{(\infty)}_i(q_j) = \shapley^{(\infty)}_i(w(P_j))$. Let $k(j)$ be the minimal $k \notin P_j$. Since $q_j \leq w(\infty)-w_{k(j)}$, $k(j) \to \infty$. Therefore we can assume without loss of generality that $k(j) > i$ for all $j$. Theorem~\ref{thm:super-increasing-limit} implies that
\[
 \shapley^{(\infty)}_i(q_j) \leq \frac{1}{i\binom{i-1}{i-1}} - \sum_{\ell=1}^{k(j)-1-i} \frac{1}{(i+\ell)\binom{i+\ell-1}{i+\ell-2}} = \frac{1}{k(j)-1},
\]
 using Lemma~\ref{lem:super-increasing-sum-identity}. Since $k(j) \to \infty$, $\shapley^{(\infty)}_i(q_j) \to 0$.
 \end{proof}

\end{document}